\documentclass[11pt]{article}
\pdfoutput=1
\usepackage[sort&compress,numbers]{natbib}
\usepackage{amsmath,amssymb,keyval,listingsmod}
\usepackage{xspace,url}
\usepackage{prologuemod,ifpdf} 

\usepackage{epsfig}
\usepackage{subfigure}
\usepackage{calc}
\usepackage{amssymb}
\usepackage{amstext}
\usepackage{amsmath}
\usepackage{amsthm}
\usepackage{multicol}
\usepackage{pslatex}

\usepackage{wrapfig}

\newcommand{\op}[1]{{\textsf{#1}}}

\newcommand{\Section}[1]{\section{#1}}
\newcommand{\Subsection}[1]{\noindent {\bf #1}~~~}
\newcommand{\Subsubsection}[1]{\noindent {\em #1}~~~}

\newcommand{\Subsubsubsection}[1]{ \noindent $\bullet$ {\bf #1}~~~~~~}

\newcommand{\Ref}[1]{}

\newcommand{\two}[4]{
  \parbox{\columnwidth}{\vspace{1pt} \vfill
    \parbox[t]{#1\columnwidth}{#3}%
    \parbox[t]{#2\columnwidth}{#4}%
  }}

\newtheorem{theorem}{Theorem}[section]
\newtheorem{lemma}[theorem]{Lemma}
\newtheorem{proposition}[theorem]{Proposition}
\newtheorem{definition}[theorem]{Definition}

\newcommand{\MAC}{\textmd{MAC}\xspace}

\newcommand{\timeslot}{timeslot\xspace}
\newcommand{\Timeslots}{Timeslots\xspace}
\newcommand{\timeslots}{timeslots\xspace}

\newcommand{\manet}{DynWANs\xspace}

\newcommand{\Manets}{DynWANs\xspace}

\newcommand{\Figure}{Fig.\xspace}
\newcommand{\Aloha}{Slotted ALOHA\xspace}

\newcommand{\dstclr}{\MAC algorithm in \Figure~\ref{f:smp}\xspace}
\newcommand{\dst}{\MAC algorithm\xspace}

\newcommand{\prp}{Self-stabilizing TDMA-based \MAC algorithm\xspace}

\newcommand{\RSONE}{\op{Ready}\xspace}
\newcommand{\RSTWO}{\op{Obtaining}\xspace}
\newcommand{\RSTHREE}{\op{Allocated}\xspace}

\begin{document}
\begin{titlepage}
\title{Self-Stabilizing TDMA Algorithms\\ for Dynamic Wireless Ad-hoc Networks~\footnote{This work was partially supported by the EC, through project FP7-STREP-288195, KARYON (Kernel-based ARchitecture for safetY-critical cONtrol).}}

\author{Pierre Leone~\footnote{Computer Science Department, University of Geneva, Centre Universitaire d'Informatique, Battelle b\^atiment A, route de Drize 7, 1227 Carouge, Geneva, Switzerland. Email: pierre.leone@unige.ch} \and Elad M.\ Schiller~\footnote{Chalmers University of Technology, R\"{a}nnv\"{a}gen 6B, S-412 96 G\"{o}teborg Sweden. Email: elad.schiller@chalmers.se}
}

\maketitle


\abstract{In dynamic wireless ad-hoc networks (\manet), autonomous computing devices set up a network for the communication needs of the moment. These networks require the implementation of a medium access control (MAC) layer. We consider MAC protocols for \manet that need to be autonomous and robust as well as have high bandwidth utilization, high predictability degree of bandwidth allocation, and low communication delay in the presence of frequent topological changes to the communication network. Recent studies have shown that existing implementations cannot guarantee the necessary satisfaction of these timing requirements. We propose a self-stabilizing MAC algorithm for \manet that guarantees a short convergence period, and by that, it can facilitate the satisfaction of severe timing requirements, such as the above. Besides the contribution in the algorithmic front of research, we expect that our proposal can enable quicker adoption by practitioners and faster deployment of \manet that are subject changes in the network topology.}


\end{titlepage}

\Section{Introduction}
\label{s:int}
Dynamic wireless ad-hoc networks (\manet) are autonomous and self-organizing systems where computing devices require networking applications when a fixed network infrastructure is not available or not preferred to be used. In these cases, computing devices may set up a short-lived network for the communication needs of the moment, also known as, an ad-hoc network. Ad-hoc networks are based on wireless communications that require implementation of a {\em Medium Access Control} (\MAC) layer. We consider MAC protocols for \manet that need to be autonomous, robust, and have high bandwidth utilization, a high predictability degree of bandwidth allocation, and low communication delay~\cite{hartenstein2010vanet} in the presence of frequent changes to  the communication network topology. Existing implementations cannot guarantee the necessary satisfaction of timing requirements~\cite{DBLP:conf/vtc/BilstrupUSB08,katrin2009ability}. This work proposes an algorithmic design for self-stabilizing \MAC protocols that guarantees a short convergence period, and by that, can facilitate the satisfaction of severe timing requirements. The proposed algorithm possesses a greater degree of predictability, while maintaining low communication delays and high throughput.

The dynamic and difficult-to-predict nature of wireless ad-hoc networks gives rise to many fault-tolerance issues and requires efficient solutions. \Manets, for example, are subject to transient faults due to hardware/software temporal malfunctions or short-lived  violations of the assumed settings for modeling the location of the mobile nodes.
Fault tolerant systems that are {\em self-stabilizing}~\cite{D00} can recover after the occurrence of transient faults, which can cause an arbitrary corruption of the system state (so long as the program's code is still intact),
%
or the model of dynamic networks in which communication links and nodes may fail and recover during normal operation~\cite{DBLP:journals/cjtcs/DolevH97}.
The proof of self-stabilization requires convergence from an arbitrary starting system state. Moreover, once the system has converged and followed its specifications, it is required to do so forever.
The self-stabilization design criteria liberate the application designer from dealing with low-level complications, such as bandwidth allocation in the presence of topology changes, and provide an important level of abstraction. Consequently, the application design can easily focus on its task -- and knowledge-driven aspects.

The {\tt IEEE 802.11} standard is widely used for wireless communications. Nonetheless, the research field of \MAC protocols is very active and requires further investigation. In fact, the {\tt IEEE 802.11} amendment, {\tt IEEE 802.11p}, for wireless access in vehicular environments (WAVE), has just being published. It was shown that the standard's existing implementations cannot guarantee channel access before a finite deadline~\cite{DBLP:conf/vtc/BilstrupUSB08,katrin2009ability}. Therefore, applications with severe timing requirements cannot predictably meet their deadlines, e.g., safety-critical applications for vehicular systems.



%
ALOHAnet and its synchronized version \Aloha~\cite{abramson1985da} are pioneering wireless systems that employ a strategy of ``random access''. Time division multiple access (TDMA)~\cite{schmidt1974std} is another early approach, where nodes transmit one after the other, each using its own \timeslot, say, according to a defined schedule. Radio transmission analysis in ad-hoc networks~\cite{Haenggi09} and relocation analysis of mobile nodes~\cite{DBLP:conf/algosensors/LeonePS09} show that there are scenarios in which \MAC algorithms that employ a scheduled access strategy have lower throughput than algorithms that follow the random access strategy. However, the scheduled approach offers greater predictability of bandwidth allocation and communication delay, which can facilitate fairness~\cite{DBLP:conf/algosensors/HermanT04} and energy conservation~\cite{DBLP:conf/infocom/YeHE02}.
%

Our design choices have basic radio technology in mind, whilst aiming at satisfying applications that have severe timing requirements. We consider TDMA frames with fixed number of fixed length \timeslots. The design choice of TDMA frames with fixed-length radio time fits well applications that have severe delay requirements. By avoiding the division of fixed length frames into \timeslots of non-equal length, as in~\cite{DBLP:conf/algosensors/HermanT04,DBLP:conf/wdag/CornejoK10}, we take into consideration the specifications of basic radio technology.

In the context of the above design choices, there are two well-known approaches for dealing with contention (\timeslot exhaustion): (1) employing policies for administering message priority (for meeting timing requirements while maintaining high bandwidth utilization, such as~\cite{DBLP:journals/cn/RomT81}), or (2) adjusting the nodes' individual transmission signal strength or carrier sense threshold~\cite{DBLP:conf/vtc/ScopignoC09}. The former approach is widely accepted and adopted by the IEEE {\tt 802.11p} standard, whereas the latter has only been evaluated via computer simulations. The proposed algorithm facilitates the implementation of both of the above approaches. We consider implementation details of the standard approach in Section~\ref{s:imp}.

For the sake of presentation simplicity, we start by considering a single priority MAC protocol and base the \timeslot allocation on vertex-coloring, before considering multi-priority implementation in Section~\ref{s:imp}. The proposed algorithm allocates \timeslots to a number of nearby transmitters, i.e., a number that is bounded by the TDMA frame size, whereas non-allocated transmitters receive busy channel indications. The analysis considers saturated situations in which the node degree in the message collision graph is smaller than the TDMA frame size. As explained above, this analysis assumption does not restrict the number of concurrent transmitters when implementing the proposed MAC algorithm.

\Subsection{Related work}
%
%
We are not the first to propose a MAC algorithm for \manet that follows the TDMA's scheduled approach.
%
%
%
STDMA~\cite{YB07} and Viqar and Welch~\cite{DBLP:conf/algosensors/ViqarW09} consider GNSS-based scheduling (Global Navigation Satellite System~\cite{DBLP:conf/ntms/ScopignoC09}) according to the nodes' geographical position and their trajectories.
Autonomous systems cannot depend on GNSS services, because they are not always available, or preferred not to be used, due to their cost. Arbitrarily long failure of signal loss can occur in underground parking lots and road tunnels.
%
%
We propose a self-stabilizing TDMA algorithm that does not require GNSS accessibility or knowledge about the node trajectories. Rather it considers an underlying self-stabilizing local pulse synchronization, such as~\cite{DBLP:conf/sss/DaliotDP03,MPSTT12}, which can be used for TDMA alignment, details appear in~\cite{MPSTT12}.

When using collision-detection at the receiving-side~\cite{DBLP:conf/vtc/ScopignoC09,CS09,YB07,TII08,lenoble2009header}, it is up to the receiving-side to notify the sender about collisions via another round of collision-prone transmissions, say, by using FI (frame information) payload fields that includes $T$ entries, where $T$ is the TDMA frame size.
Thus far, the study of FI-based protocols has considered stochastic resolution of message collision via computer network simulation~\cite{YB07,DBLP:conf/vtc/AbrateVS11,DBLP:conf/vtc/ScopignoC10,DBLP:conf/apscc/CozzettiSCB09,TII08,lenoble2009header}.
Simulations are also used for evaluating the heuristics of MS-ALOHA~\cite{DBLP:conf/vtc/ScopignoC09} for dealing with contention (timeslot exhaustion) by adjusting the nodes' individual transmission signal strength and / or carrier sense threshold.
%
We do not consider lengthy frame information (FI) fields, which significantly increase the control information overhead, and yet we provide provable guarantee regarding the convergence time. Further analysis validation of the proposed algorithm via simulations and test-bed implementation can be found in Section~\ref{s:dis}, and respectively, in~\cite{MPSTT12}.

The proposed algorithm does {\em not} consider collision-detection mechanisms that are based on signal processing or hardware support, as in~\cite{DBLP:conf/broadnets/DemirbasH06}. Rather, it employs a variation on a well-known strategy for eventually avoiding concurrent transmissions among neighbors. This strategy allows the sending-side to eventually observe the existence of interfering transmissions. Before sending, the sender waits for a random duration while performing a clear channel assessment using basic radio technology. We assume that, in the presence of a nearby transmission, the radio unit can detect that the energy level has reached a threshold in which the radio unit is expected to succeed in carrier sense locking (details appear in Section~\ref{s:alg}).

Another MAC algorithm that bases its clear channel assessment on carrier sensing of message transmission appears in~\cite{DBLP:conf/wdag/CornejoK10}. The authors focus on fair bandwidth allocation for single-hop-distance broadcasting, but do not consider dynamic networks or self-stabilization.
Both the algorithm in~\cite{DBLP:conf/wdag/CornejoK10} and the proposed one can facilitate unicast functionality with no significant overhead.
By basing the clear channel assessment on carrier sensing, we mitigate the effect of transmission pathologies, such as hidden terminal phenomena, and reduce the need for additional mitigation efforts, such as self-stabilizing two-hop-distance vertex
coloring~\cite{DBLP:conf/sirocco/BlairM09}, equalizing transmission power, coding-based methods~\cite{DBLP:conf/sigcomm/GollakotaK08}, to name a few.


An abstract \MAC layer was specified for \manet in~\cite{DBLP:conf/wdag/KuhnLN09}. The authors mention algorithms that can satisfy their specifications. However, they do not consider predictability.

{\em Local algorithms}~\cite{DBLP:conf/icalp/HalldorssonW09,DBLP:conf/infocom/GoussevskaiaWHW09} considers both theoretical and practical aspects of \MAC algorithms~\cite[][and references therein]{Wattenhofer2010Theory} and the related problem of clock synchronization, see~\cite{Lenzen2010Clock} and references therein. For example, the first partly-asynchronous self-organizing local algorithm for vertex-coloring in wireless ad-hoc networks is presented in~\cite{DBLP:conf/podc/SchneiderW09}. However, this line currently does not consider dynamic networks and predictable bandwidth allocation.

%
%

%
Two examples of self-stabilizing TDMA algorithms are presented in~\cite{DBLP:conf/algosensors/HermanT04,DBLP:conf/icdcit/JhumkaK07}. The algorithms are based on vertex-coloring and the authors consider (non-dynamic) ad-hoc networks. Recomputation and floating output techniques (\cite{D00}, Section 2.8) are used for converting deterministic local algorithms to self-stabilization in~\cite{DBLP:conf/sss/LenzenSW09}. The authors focus on problems that are related to \MAC algorithms. However, deterministic \MAC algorithms are known to be inefficient in their bandwidth allocation when the topology of the communication network can change frequently~\cite{DBLP:conf/algosensors/LeonePS09}. There are several other proposals related to self-stabilizing \MAC algorithms for sensor networks, e.g.,~\cite{kulkarni-ss,DBLP:conf/icdcit/ArumugamK05,arumugam2006self,Telematik_LNWT_2009_SelfWISE}; however, none of them consider dynamic networks and their frame control information is quite extensive.



The MAC algorithms in~\cite{DBLP:conf/algosensors/LeonePS09,DBLP:conf/sss/LeonePSZ10,LPS09APMAC,MPSTT12} {\em have no proof} that they are self-stabilizing. The authors of~\cite{DBLP:conf/algosensors/LeonePS09} present a MAC algorithm that uses convergence from a random starting state (inspired by self-stabilization).
In~\cite{DBLP:conf/sss/LeonePSZ10,LPS09APMAC,MPSTT12}, the authors use computer network simulators for evaluating self-$\star$ MAC algorithms. 

\Subsection{Our contribution}
This work proposes a self-stabilizing MAC algorithm that demonstrates rapid convergence without the extensive use of frame control information. Our analysis shows that the algorithm facilitates the satisfaction of severe timing requirements for \manet.

We start by considering transient faults and topological changes to the communication network, i.e., demonstrating self-stabilization in Theorem~\ref{th:main}. We then turn to focus on bounding the algorithm's convergence time after an arbitrary and unbounded  finite sequence of transient faults and changes to the network topology. Theorem~\ref{th:expS} shows that the expected local convergence time is brief, and bounds it in equation~$(\ref{boundht})$. Theorem~\ref{th:thbounds} formulates the expected global convergence time in equation~$(\ref{boundconv})$. Moreover, for a given probability, the global convergence time is calculated in equation~$(\ref{boundk})$.

For discussion (Section~\ref{s:dis}), we point out the algorithm's ability to facilitate the satisfaction of severe timing requirements for \manet. Moreover, the analysis conclusions explain that when allowing merely a small fraction of the bandwidth to be spent on frame control information and when considering any given probability to converge within a bounded time, the proposed algorithm demonstrates a low dependency degree on the number of nodes in the network (as depicted by \Figure~\ref{fig:bounds} and \Figure~\ref{fig:boundk}).

We note that some of the proof details appear in the Appendix for the sake of presentation simplicity. 

\Section{Preliminaries}
\label{s:pre}
The system consists of a set, $P$, of $N$ anonymous communicating entities, which we call {\em nodes}. Denote every node $p_i \in P$ with a unique index, $i$. 

\Subsection{Synchronization}
Each node has fine-grained, real-time clock hardware.
We assume that the \MAC protocol is invoked periodically by synchronous {\em (common) pulse} that aligns the starting time of the TDMA frame. This can be based, for example, on TDMA alignment algorithms~\cite{MPSTT12}, GPS~\cite{hofmann1993global} or a distributed pulse synchronization algorithm~\cite{DBLP:conf/sss/DaliotDP03}. The term {\em (broadcasting) \timeslot} refers to the period between two consecutive common pulses, $t_x$ and $t_{x+1}$, such that $t_{x+1}=(t_x \bmod T) +1$, where $T$ is a predefined constant named the {\em frame size}. Throughout the paper, we assume that $T \ge 2$. In our pseudo-code, we use the event \op{\timeslot}$(t)$ that is triggered by the common pulse. We assume that the \timeslots are aligned as well.

\Subsection{Communications and interferences}
At any instance of time, the ability of any pair of nodes to communicate is defined by the set, $N_i \subseteq P$, of {\em (direct) neighbors} that node $p_i \in P$ can communicate with directly. Wireless transmissions are subject to interferences (collisions). We consider the potential of nodes to interfere with each other's communications. The interference model in this paper is based on discrete graphs.

The set $\mathcal{N}_i \supseteq N_i$ is the set of nodes that may interfere with $p_i$'s communications when any nonempty subset of them, $I \subseteq \mathcal{N}_i : I \neq \emptyset$, transmit concurrently with $p_i$. We call $\mathcal{N}_i$ the {\em (extended) neighborhood} of node $p_i \in P$ and ${d_i}=|\mathcal{N}_i|$ is named the (extended) degree of node $p_i$. We assume that at any time, for any pair of nodes, $p_i, p_j \in P$ it holds that $p_j \in \mathcal{N}_i$ implies that $p_i \in \mathcal{N}_j$. Given a particular instance of time, we define the {\em (interference) graph} as $G :=(P, E)$, where $E :=\cup_{i \in P} \{ (p_i, p_j) : p_j \in \mathcal{N}_i \}$ represents the interference relationships among nodes.


\Subsection{Communication schemes}
We consider (basic technology of) radio units that raise the event \op{carrier$\_$sense}$()$ when they detect that the received energy levels have reached a threshold in which the radio unit is expected to succeed in carrier sense locking, see~\cite{jamieson2005understanding}.
\Timeslots allow the transmission of \op{DATA} packets using the \op{transmit}$()$ and \op{receive}$()$ primitives after fetching  (\op{MAC}$\_$\op{fetch}$()$) a new packet from the upper layer, and respectively, before delivering (\op{MAC}$\_$\op{deliver}$()$) the packet to the upper layer. A {\em beacon} is a short packet that includes no data load, rather the timing of the event \op{carrier$\_$sense}$()$ is the delivered information~\cite{DBLP:conf/wdag/CornejoK10}. Before the transmission of the \op{DATA} packet in \timeslot $t$, our communication scheme uses beacons for signaling the node's intention to transmit a \op{DATA} packet within $t$, see \Figure~\ref{f:frm}. 


\begin{figure}[t]
\begin{flushright}
\begin{minipage}{0.95\textwidth}
\ifx\pdftexversion\undefined
\includegraphics[viewport=0 0 722 142,clip,scale=0.85]{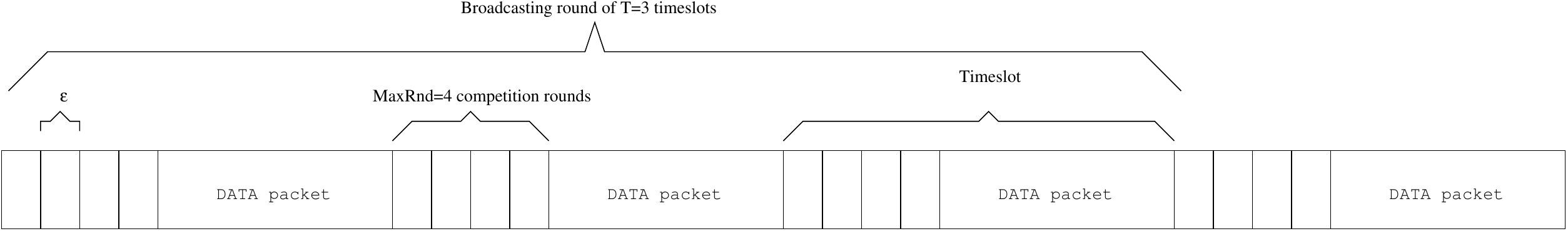}
\else
{\includegraphics[viewport=0 0 380 70,clip,scale=0.85]{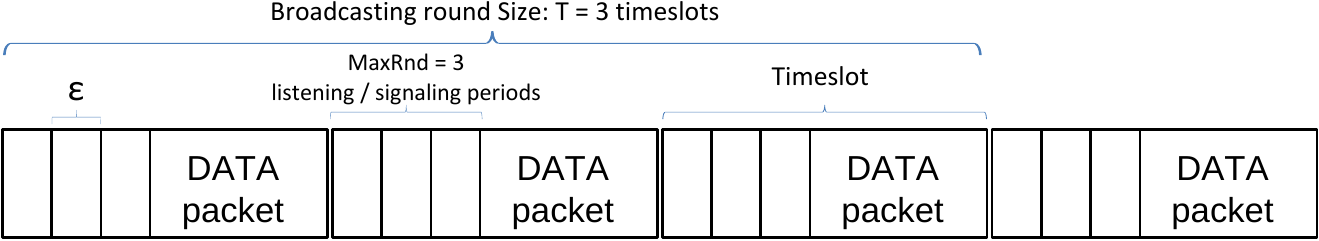}}
\fi
\end{minipage}
\end{flushright}
\caption{An example of TDMA frame, with three \timeslots and three listening/signaling periods of size $\epsilon$ (signal exposure time). Each \timeslot has a constant number, $MaxRnd=4$, of {\em listening/signaling periods} in which beacons can be transmitted. The duration of each listening/signaling period is $\epsilon$ (signal exposure time); the period during which a beacon that is sent by node $p_i \in P$ is transmitted and raises the ca received by all neighbors $p_j \in \mathcal{N}_i$. Namely, the period between $p_i$'s transmission and $p_j$'s rise of the \op{carrier$\_$sense}$()$ event.
\label{f:frm}}
\end{figure}

The interference model in this paper is based on discrete graphs and the clear channel assessment is based on carrier sensing of message transmission. By basing the clear channel assessment on carrier sensing, we mitigate the effect of transmission pathologies, such as hidden terminal phenomena, and reduce the need for additional mitigation efforts, such as self-stabilizing two-hop-distance vertex coloring~\cite{DBLP:conf/sirocco/BlairM09}, equalizing transmission power, coding-based methods~\cite{DBLP:conf/sigcomm/GollakotaK08}, to name a few. We note that we are not the first to assume that every node $p_i \in P$ that invokes the operation \op{transmit}$()$ causes the event \op{carrier$\_$sense}$()$ to be raised by its neighbors, $p_j \in \mathcal{N}_i$, within the exposure time, $\epsilon$~\cite{DBLP:conf/wdag/CornejoK10}.


\Subsection{System settings}
\label{s:sys}
%
We consider the interleaving model~\cite{D00}. Every node, $p_i \in P$, executes a program that is a sequence of {\em atomic steps}. The {\em state} $st_i$ of a node $p_i$ consists of the value of all the variables of the node (including messages in transit for $p_i$). Variables are associated with individual node states by using the subscript notation, i.e., $var_i$ is the value of variable $var$ in $p_i$'s state. The term {\em configuration} is used for a tuple of the form $(G, \{ st_i \}_{i=1}^{{N}})$, where $G$ is the (interference) graph, and $\{ st_i \}_{i=1}^{{N}}$ are the nodes' states (including the set of all incoming communications). An {\em execution} (run) $R :=(c(0),c(1),\ldots)$ is an unbounded sequence of system configurations $c(x)$, such that each configuration $c(x+1)$ (except the initial configuration $c(0)$) is obtained from the preceding configuration $c(x)$ by the execution of steps, $\{ a_i(x) \}_{p_i \in P}$, taken by all nodes.

Let $\tau$ (task) be a specification set and $LE$ a set of all executions that satisfy task $\tau$. We consider TDMA-based \MAC protocols for which the task $\tau_{_\mathrm{TDMA}}$ considers requirements in which every node has its own broadcasting \timeslot that is unique within its neighborhood. We note that $\tau_{_\mathrm{TDMA}}$'s requirements are obviously satisfiable when the ratio between the extended degree and the frame size is less than one, i.e., there is no \timeslot exhaustion when $\forall p_i \in P : 1 \lneq T/d_i$. The system deals with \timeslot exhaustion by delivering busy channel indications, $\bot$, to the nodes for which there was no \timeslot left. We define $LE_{_\mathrm{TDMA}}$ to be the set of legal executions, $R$, for which $\forall p_i \in P : (((s_i \in [0, T-1]) \wedge (p_j \in {\cal N}_i)) \Rightarrow s_i \neq s_j) \vee (s_i = \bot \Rightarrow \forall t \in [0, T-1]~ \exists p_j \in {\cal N}_i : s_j = t)$ in all of $R$'s configurations.

We say that configuration $c_{safe}$ is {\em safe} if there is an execution $R \in LE$, such that $c_{safe}$ is $R$'s starting configuration. Let $R$ be an execution and $c \in R$ its arbitrary starting configuration. We say that $R$ {\em converges} with respect to $\tau$ if within a bounded number of steps from $c$, the system reaches a {\em safe} configuration $c_{safe}$. The {\em closure} property requires that for any execution, $R$, that starts form $c_{safe}$ implies that $R \in LE$. An algorithm is said to be {\em self-stabilizing} if it satisfies both the convergence and the closure properties.

We describe execution $R$ as an unbounded number of concatenated finite sequences of configurations. The finite sequences, $R(x)=(c{^{_{\op{}}}_{^{0}}}(x),\ldots$ $c{^{_{\op{}}}_{^{T-1}}}(x))$, $x>0$, is a {\em broadcasting round} if (1) configuration $c{^{_{\op{}}}_{^{0}}}(x)$ has a clock value, $t$, of $0$ and
immediately follows a configuration in which the clock value is $T-1$, and (2) configuration $c{^{_{\op{}}}_{^{T-1}}}(x)$ has a clock value of $T-1$ and
immediately precedes a configuration in which the clock value is $0$.


\begin{figure*}[t!]
\begin{flushleft}
\begin{minipage}[b]{0.99\linewidth}
\begin{center}
\fbox{\hspace*{+10pt}{
\begin{lstlisting}[language=ioaNums]
*Constants, variables, macros and external functions*
$~~$MaxRnd ($n$ in the proofs) : integer = bound on round number  $\label{l:max}$
$~~$s : [0, T-1] \cup {\bot} = next timeslot to broadcast $or$ null, \bot
$~~$signal : boolean = trying to acquire the channel
$~~$unused[0,T-1] : boolean = marking unused timeslots          $\label{l:emptydef}$
$~~$unused$\_$set = { k : unused[k] = true } : unused timeslots, macro $\label{l:emptyset}$
$~~$!MAC$\_$fetch!()/!MAC$\_$deliver!() : MAC layer interface
$~~$!transmit!/!receive!/!carrier$\_$sense! : communication primitives

*Upon* !timeslot!(t)
$~~$if t = 0 /\ s = \bot$~$then s := !select_unused!(unused$\_$set) $\label{l:slcemp}$   
$~~$(unused[t], signal) := (true, false) $(*$ remove stale info. $*)$  $\label{l:outdated}$
$~~$if s \ne \bot /\ t = s then !send!(!MAC$\_$fetch!())            $\label{l:fetch}$

*Upon* !receive!(<DATA, m>) do !MAC$\_$deliver!(<m>) $\label{l:receive}$

*Function* !send!(m) $~~$$(*$ send message m to $p_i^\prime$s neighbors $*)$
$~~$for ((signal, k) := (true, 0); k := k + 1; k \le MaxRnd) do                          $\label{l:for}$
$~~$$~~$if $signal$ then with probability ${\rho(k)}=1/(MaxRnd-k)$ do                    $\label{l:p}$
$~~$$~~$$~~$signal := false $(*$ quit the competition $*)$                               $\label{l:quit}$
$~~$$~~$$~~$!transmit!(<BEACON>) $(*$ try acquiring the channel $*)$                     $\label{l:starta}$
$~~$$~~$!wait! until the end of competition round $(*$ exposure period alignment $*)$    $\label{l:wait}$
$~~$if $s \ne \bot$ then !transmit!(<DATA, m>) $(*$ send the data packet $*)$            $\label{l:transmit}$

*Upon* !carrier$\_$sense!(t) $(*$ defer transmission during t $*)$ $\label{l:remotea}$
$~~$if s = t /\ signal then s := $\bot$ $(*$ mark that the timeslot is $not$ unique $*)$       $\label{l:state}$
$~~$(signal, unused[t]) := (false, false)  $(*$ quit the competition $*)$          $\label{l:emptr}$

*Function* !select_unused!(set)  $(*$ select an empty timeslot $*)$                      $\label{l:uniformselect}$
$~~$if $set = \emptyset$ then return $\bot$ else return $uniform\_select(set)$           $\label{l:retunrandom}$


\end{lstlisting}}}
\end{center}
\end{minipage}
\end{flushleft}
\caption[]{\prp, code of node $p_i$.}\label{f:smp}
\end{figure*}

\Section{Algorithm Description}
\label{s:alg}
The proposed MAC algorithm periodically performs clear channel assessment and uses this assessment when informing each node about the nearby unused \timeslots. The nodes use this information for selecting their broadcasting \timeslots, assess the success of their  broadcasts and reselecting \timeslots when needed.

The \dstclr satisfies the $\tau_{_\mathrm{TDMA}}$ task. 
During the convergence period several nodes can be assigned to the same \timeslot. Namely, we may have $p_i \in P : p_j \in {\cal N}_i \wedge s_i = s_j$. The algorithm solves such \timeslot allocation conflicts by letting the node $p_i$ and $p_j$ to go through a (listening/signaling) competition before transmitting in its broadcasting \timeslot. The competition rules require each node to choose one out of $MaxRnd$ listening/signaling period for its broadcasting \timeslot, see \Figure~\ref{f:frm}. This implies that among all the nodes that attempt to broadcast in the same \timeslot, the ones that select the earliest listening/signaling period win this broadcasting \timeslot and access the communication media. Before the winners access their \timeslots, they signal to their neighbors that they won via {beacon} transmission. The signal is sent during their choice of listening/signaling periods, see \Figure~\ref{f:frm}. When a node receives a {beacon}, it does not transmit during that \timeslot, because it lost this (listening/signaling) competition. Instead, it randomly selects another broadcasting \timeslot and competes for it on the next broadcasting round.

In detail, the \dstclr is invoked at the start of every \timeslot, $t$. When $t$ is the first \timeslot, the algorithm tries to allocate the broadcasting \timeslot, $s_i$, to $p_i$ (line~\ref{l:slcemp}) by randomly selecting a \timeslot for which there is no indication to be used by its neighbors. Later, when the \timeslot $t$ becomes $p_i$'s broadcasting \timeslot, $s_i$, the node attempts to broadcast (by calling the function \op{send}$()$ in line~\ref{l:fetch}). We note that the start of \timeslot $t$ also requires the marking of $t$ as an unused \timeslot and the removal of stale information (line~\ref{l:outdated}). This indication is changed when the $\op{carrier$\_$sense}(t)$ event is raised (line~\ref{l:emptr}) due to a neighbor transmission. Namely, when the detected energy levels reach a threshold in which the radio unit is expected to succeed in carrier sense locking, see~\cite{jamieson2005understanding}.

When a node attempts to broadcast it uses the (listening/signaling) competition mechanism for deciding when to signal to its neighbors that it is about to transmit a \op{DATA} packet. The competition has $MaxRnd$ rounds and it stops as soon as the node transmits a {beacon} or a neighbor succeeds in signaling earlier (lines~\ref{l:for} to~\ref{l:transmit}). We note that this signaling is handle by the $\op{carrier$\_$sense}(t)$ event (line~\ref{l:emptr}). Moreover, {beacons} are not required to carry payloads or any other information that is normally stored in packet headers. They are rather used to invoke the carrier sense event in ${\cal N}_i$.

The carrier sense in \timeslot $t$ indicates to each node that it needs to defer from transmission during $t$ (line~\ref{l:remotea}). In particular, it should stop using \timeslot $t$ for broadcasting, stop competing and mark $t$ as a used \timeslot. Lastly, arriving \op{DATA} packets are delivered to the upper layer (line~\ref{l:receive}).

\Section{Correctness Proof: Outline and Notation}
\label{s:out}
%
%
%
The proof starts by considering networks that do not change their topology and for which the ratio between the extended degree of node and the frame size is less than one, i.e., $\forall p_i \in P : 1 \lneq T/d_i $. For these settings, we show that the \dstclr is self-stabilizing with respect to task $\tau_{_\mathrm{TDMA}}$ (sections~\ref{s:invarI} to~\ref{s:invarII} of the Appendix), before considering the converge time within a single neighborhood (Section~\ref{s:markov}) and the entire neighborhood (Section~\ref{s:global}). These convergence estimations facilitate the exploration of important properties, such as predictability, and dealing with changes in the network topology of \manet (Section~\ref{s:dis}).

\Subsection{Proof outline}
The exposition of the proof outline refers to Definition~\ref{def:allocated}, which delineates the different states at which a node can be in relation to its neighbors. Definition~\ref{def:allocated} groups these states into three categories of {\em relative states}: (1) \RSONE to be allocated, when the node state depicts correctly its neighbor states, (2) \RSTWO a \timeslot, when the node is competing for one, but there is no agreement with its neighbor states, and (3) \RSTHREE to a \timeslot, when the node is the only one to be allocated to a particular \timeslot in its neighborhood. The correctness proof shows that the \dstclr implements $\tau_{_\mathrm{TDMA}}$ in a self-stabilizing manner by showing that eventually all nodes are allocated with \timeslots, i.e., all nodes are in the relative state \RSTHREE, see Definition~\ref{def:allocated}.

Let $R$ be an execution of the \dstclr and $R(x)$ is the $x$-th complete broadcasting round of $R$, where $x>0$ is an integer. We simplify the presentation by using uppercase notation for the configurations, $c{^{_{\op{name}}}_{^{t}}}(x)$, where $t \in [0, T-1]$ is a \timeslot. This notation includes the $name$ of the first event to be triggered immediately after configuration $c$, i.e., $R(x)=(c{^{_{\op{timeslot}}}_{^{0}}}(x), \ldots$ $c{^{_{\op{carrier\_sense/receive}}}_{^{T-1}}}(x))$.

\begin{definition}
\label{def:allocated}
We say that node $p_i \in P$ is {\em \RSONE (to be allocated)} to a \timeslot in configuration $c{^{_{\op{timeslot}}}_{^{0}}}(x)$, if properties~$(\ref{eq:invarIa})$,~$(\ref{eq:invarIb})$ and~$(\ref{eq:invarIIai})$ hold for node $p_i$ but Property~$(\ref{eq:invarIIa})$ does not.
We say that $p_i$ is {\em \RSTWO} \timeslot $s_i$ in configuration $c{^{_{\op{timeslot}}}_{^{0}}}(x)$, if properties~$(\ref{eq:invarIa})$ to~$(\ref{eq:invarIIa})$ hold for node $p_i$, but Property~$(\ref{eq:invarIIb})$ does not.
We say that node $p_i \in P$ is in {\em \RSTHREE} state, with respect to \timeslot $s_i$ in configuration $c{^{_{\op{timeslot}}}_{^{0}}}(x)$, if properties~$(\ref{eq:invarIa})$ to~$(\ref{eq:invarIIb})$ hold for node $p_i$.

\begin{equation}
      \label{eq:invarIa}
      signal_i=\op{false}
\end{equation}
\begin{equation}
    \label{eq:invarIb}
    ( t \in unused_i \wedge t \neq s_i ) \leftrightarrow (\forall p_k \in {\cal N}_i : s_k \neq t)
\end{equation}
\begin{equation}
      \label{eq:invarIIai}
  s_i \neq \bot \vee unused\_set_i \setminus \{ s_i \} \neq  \emptyset
\end{equation}
\begin{equation}
      \label{eq:invarIIa}
  s_i  \neq  \bot
\end{equation}
\begin{equation}
      \label{eq:invarIIb}
       \forall p_j \in {\cal N}_i : ((s_i \neq s_j) \wedge (unused_j[s_i]=\op{false}))
\end{equation}
\end{definition}

Property~$(\ref{eq:invarIa})$ implies that node $p_i$ finishes any broadcast attempts within a \timeslot. Properties~$(\ref{eq:invarIb})$ to~$(\ref{eq:invarIIai})$ consider the case in which $p_i$'s internal state represents correctly the \timeslot allocation in its neighborhood. In particular, property~$(\ref{eq:invarIb})$ means that processor $p_i$ views timeslot $t$ as an unused one if, and only if, it is indeed unused. Property~$(\ref{eq:invarIIai})$ implies that when node $p_i$ is not using any \timeslot, there is an unused \timeslot at its disposal. Property~$(\ref{eq:invarIIa})$ says that node $p_i$ is using \timeslot $s_i$. Property~$(\ref{eq:invarIIb})$ refers to situations in which $p_i$'s neighbors are not using $p_i$'s \timeslot during the next broadcasting round.

Starting from an arbitrary configuration, we show that node $p_i$ becomes \RSONE within two broadcasting rounds (or one complete broadcasting round), see Section~\ref{s:invarI} of the Appendix. Then, we consider the probability, $OnlyOne_i(x)$, that a node enters the relative state \RSTHREE from either \RSONE or \RSTWO, see equation~$(\ref{equ:qi})$ (and sections~\ref{s:invarII} and~\ref{s:compact} of the Appendix). Namely, equation~$(\ref{equ:qi})$ considers the probability that node $p_i$ is the {\em only one} to use its broadcasting timeslot in its neighborhood, where ${{\rho}_k}=1/MaxRnd=1/n$ is $p_i$'s probability to selects the $k$-th listening/signaling period for transmitting its {beacon}.

\begin{equation}
\label{equ:qi}
\boxed{OnlyOne_i({x})\ge\sum_{k=1}^n {\rho}_k \left( 1- \sum^{k}_{\ell=1} {\rho}_k \right)^{\frac{d_i}{T}}}
\end{equation}

\Subsubsection{Demonstration of self-stabilization}
%


\begin{theorem}
\label{th:main}
The \dstclr is self-stabilizing with respect to the task $\tau_{_\mathrm{TDMA}}$.
\end{theorem}



\begin{proof}

The proofs of the propositions of this theorem appear in sections~\ref{s:invarI} and~\ref{s:invarII} of the Appendix.

\Subsubsubsection{Convergence}
We need to show that properties~$(\ref{eq:invarIa})$ to~$(\ref{eq:invarIIb})$ eventually hold in configuration $c{^{_{\op{timeslot}}}_{^{0}}}(x+y)$ for a finite value of $y>0$. Propositions~\ref{l:invarI},~\ref{l:invarI1} and~\ref{l:invarIIb} imply that properties~$(\ref{eq:invarIa})$,~$(\ref{eq:invarIb})$, and respectively,~$(\ref{eq:invarIIai})$ within two broadcasting round.

Propositions~\ref{l:invarI2},~\ref{l:invarII1} and~\ref{l:invarII2} show that there is a nonzero probability that node $p_i$ enters the relative state \RSTHREE from either \RSONE or \RSTWO within one broadcasting round. Thus, by the analyzing the expected time of the scheduler-luck games~\cite{D00,DBLP:journals/tse/DolevIM95}, we have $y$ has a finite value. Further analysis  of $y$ appears in theorems~\ref{th:expS} and~\ref{th:thbounds}.

\Subsubsubsection{Closure}
Suppose that $c{^{_{\op{timeslot}}}_{^{0}}}(x) \in R$ is a safe configuration and let $p_i \in P$ be any node. By the assumption that $c{^{_{\op{timeslot}}}_{^{0}}}(x)$, we have that $p_i$ is in the relative state \RSTHREE, i.e., properties~$(\ref{eq:invarIa})$ to~$(\ref{eq:invarIIb})$ hold for any node $p_i$. We need to show that properties~$(\ref{eq:invarIa})$ to~$(\ref{eq:invarIIb})$ holds in configuration $c{^{_{\op{timeslot}}}_{^{0}}}(x+1)$.

Propositions~\ref{l:invarI},~\ref{l:invarI1} and~\ref{l:invarIIb} imply that properties~$(\ref{eq:invarIa})$,~$(\ref{eq:invarIb})$, and respectively,~$(\ref{eq:invarIIai})$ (within one complete broadcasting round).

Properties~$(\ref{eq:invarIIa})$ to~$(\ref{eq:invarIIb})$ are implied by Proposition~\ref{l:invarII2} and the fact that Properties~$(\ref{eq:invarIIa})$ to~$(\ref{eq:invarIIb})$ holds in $c{^{_{\op{timeslot}}}_{^{0}}}(x)$, i.e., $M(x)=\emptyset$.
\end{proof}


\Subsubsection{Bounding the convergence time}
%
%
We bound the time it takes the \dstclr to converge by considering the relative states, \RSONE, \RSTWO, and \RSTHREE, and describe a state machine of a Markovian process. This process is used for bounding the convergence time of a single node (Section~\ref{s:markov}), and the entire network (Section~\ref{s:global}).

\begin{figure*}[t!]

\begin{center}
\fbox{
\begin{minipage}[b]{0.95\textwidth}
%
\FFF
\begin{wrapfigure}{l}{0.18\textwidth}\vspace*{-0.750cm}{\hspace*{-0.0cm}\vspace*{-0.0cm}\includegraphics[scale=0.26, clip, viewport=113 9 363 511]{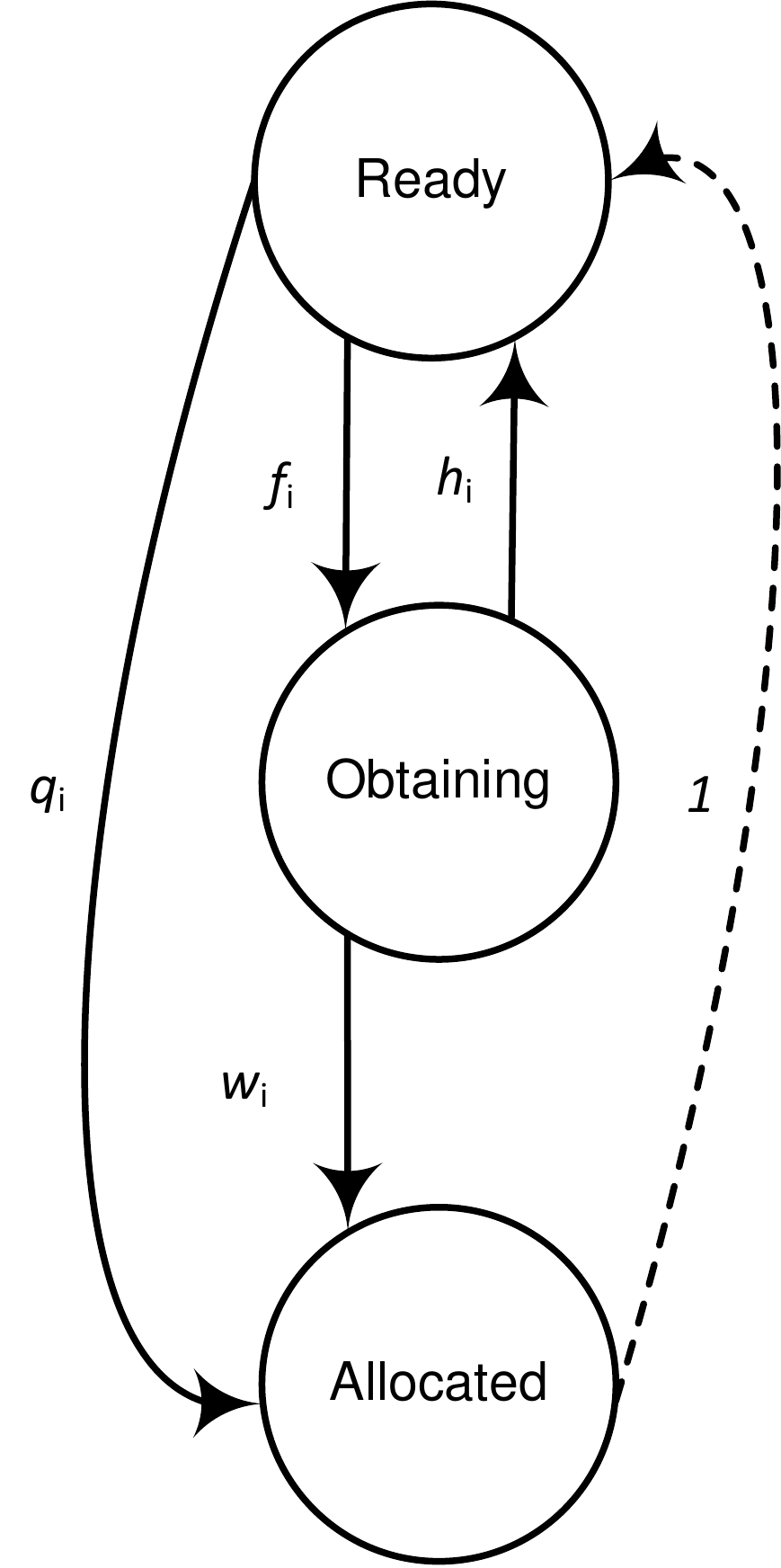}}
\end{wrapfigure}

We look at $p_i$'s state transition with relation to its neighbors, see Definition~\ref{def:allocated}. The figure on the right defines $p_i$'s relative states as a $3$-state Markov chain. The probabilities, $q_i$, $w_i$, $f_i$, and $h_i$ (solid lines arrows), that node $p_i$ change its relative state depends on its neighbor's state. For instance, $q_i$ is the probability that $p_i$ goes from the relative state \RSONE to \RSTHREE. It is environment dependent, i.e., the states of $p_i$'s neighbors are random as well. We add the dotted edge between the state \RSTHREE and the state \RSONE in order to make the Markov chain irreducible and to allow working with the invariant probability. Namely, once node $p_i$ arrives to \RSTHREE, it returns to \RSONE with probability $1$. With this convention, we can estimate the expected time to reach the final relative state \RSTHREE from relative state \RSONE by the expectation of the first hitting time of the irreducible chain~\cite{aldous}
\FFF
\end{minipage}
} 
\end{center}

\begin{center}
\caption{Markov chain describing $p_i$'s relative state transitions.}\label{fig:markov}
\end{center}
%
%
\end{figure*}

In detail, give node $p_i \in P$, its neighborhood, ${\cal N}_i$, we define a random environment of a Markov chain, see \Figure~\ref{fig:markov}. By looking at this random environment, we can focus our analysis on $p_i$'s relative states while avoiding probability dependencies and considering average probabilities~\cite{cogburn84}. Suppose that $p_i$'s environment, $e$, is known. Theorem~\ref{th:expS} estimates two bounds on the expectation of probability, $q_i \mid_e$, which is literally the probability $q_i$ given that the environment is $e$.

In order to do that, we consider a set, ${\cal R}$, of executions of the \dst, such that each execution $R \in {\cal R}$ starts in a configuration, $c \in R$, in which: (I) for any node $p_j \in P$, properties~$(\ref{eq:invarIa})$,~$(\ref{eq:invarIb})$ and~$(\ref{eq:invarIIai})$ holds, and (II) node $p_i$ is in the relative state \RSONE, which implies that (III) eventually, node $p_i$ arrives to the relative state \RSTHREE.

With this convention, we can add a probability $1$ to transit from the relative state \RSTHREE to \RSONE, see the dashed line in the state-machine diagram of \Figure~\ref{fig:markov}. This allows us to estimate the expected time to reach the final relative state \RSTHREE from relative state \RSONE by the expectation of the first hitting time of the irreducible Markov chain~\cite{aldous}.

When computing the expected time for node $p_i$ to reach state \RSTHREE within its neighborhood, we see that it is sufficient to consider the lower bound of the probability $OnlyOne_i(x)$ to obtain an upper bound on the expected time to converge, see section \ref{s:markov}. Moreover, when considering the network convergence time, i.e., the expected convergence time of all nodes in the network, we see that the most dominant parameter is the mean neighborhood size. We do that by applying the AM-GM (Arithmetic Mean vs Geometric Mean) inequality and bounding the expected network convergence time, see Section \ref{s:global}.

\Subsection{Notation}
Throughout the paper, we denote the states of the Markov chain by $\{ X_t \}_{t \ge 0}$, $T_i^+ = \min \{ t > 0 \text{ such that } X_t = i \}$ and ${E_i} \left( \cdot \right)$ is the expectation given that we start in relative state $i$, ${E_i} \left( T_i^+ \right) = E \left( T_i^+ \mid X_0=i \right)$. In this paper, the states $1$, $2$, and $3$ of the Markovian process correspond respectively to states \RSONE, \RSTWO and \RSTHREE, and the time $t=0,1,\ldots$ corresponds to configuration $c_0^{timeslot}(x+t) \in R(x+t)$, where $R(x)$ is the first complete broadcasting round in $R$ that starts in a configuration, $c_0^{timeslot}(x)$, in which all nodes are in the relative state \RSONE. For example, $E_3\left(T_3^+\right)$ is the expected time to reach the \RSTHREE state.

Let $p_i \in P$ be a node for which $s_i \neq \bot \wedge \exists p_j \in {\cal N}_i : s_j = s_i$ in configuration $c{^{_{\op{\timeslot}}}_{^{0}}}(x)$. We define ${M_i}(x) =  \{ p_j \in {\cal N}_i : s_i = s_j \}$ to be the set of $p_i$'s (broadcasting \timeslot) {\em matching} neighbors, which includes all of $p_i$'s neighbors that, during broadcasting round $R(x)$, are attempting to broadcast in $p_i$'s \timeslot. In our proofs, we use $n$ as the number of {\em listening/signaling periods}, $MaxRnd$.

\begin{figure*}[t!]
%
%
%
\begin{tabular}{ccc}
\begin{minipage}[b]{0.33\linewidth}
\subfigure{
{\includegraphics[scale=0.3, viewport=110 275 470 565]{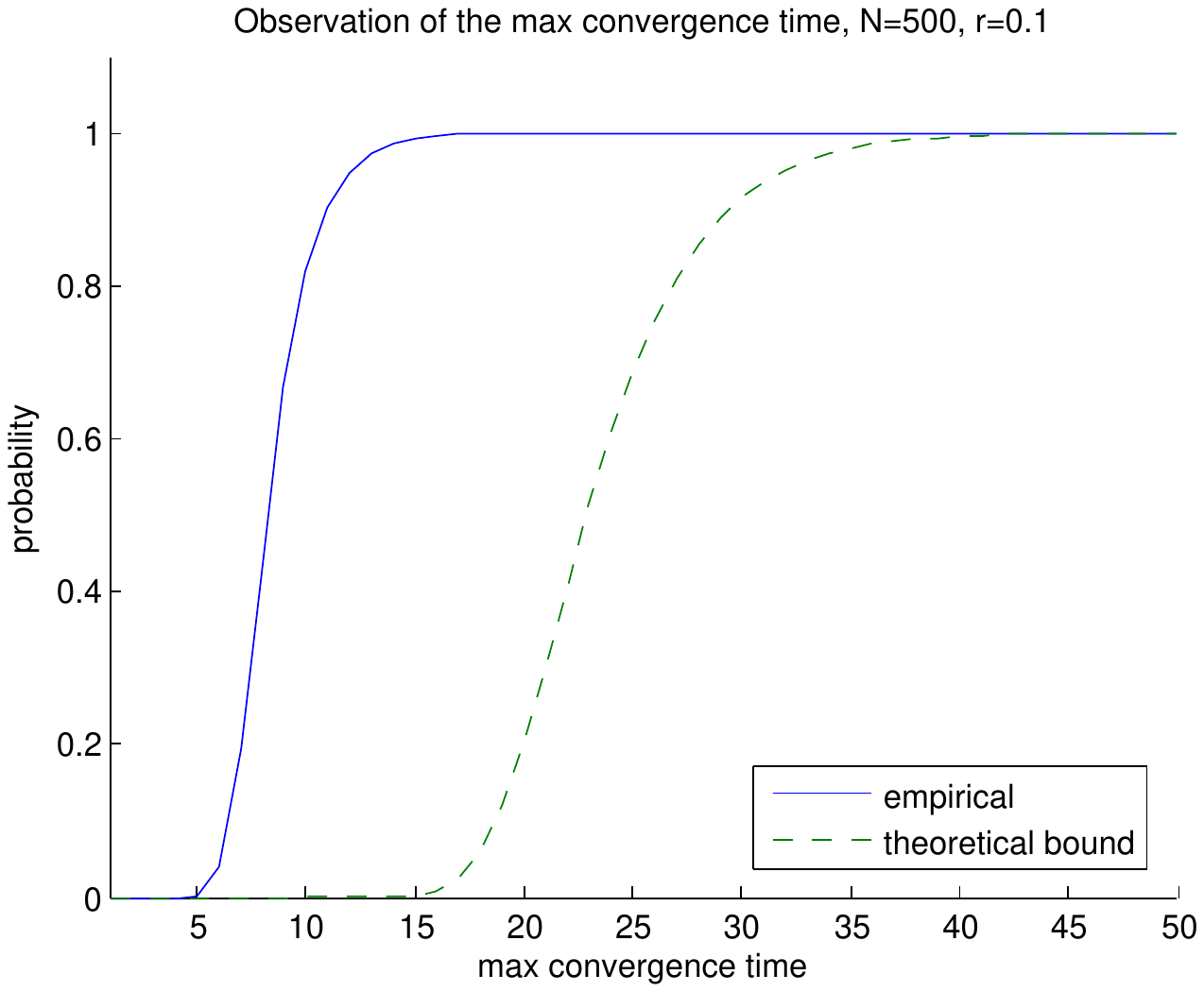}}
\label{fig:boundssubfig1}
}
\end{minipage}
&
\hskip -0.50cm
\begin{minipage}[b]{0.33\linewidth}
\subfigure{
{\includegraphics[scale=0.3, viewport=110 275 470 565]{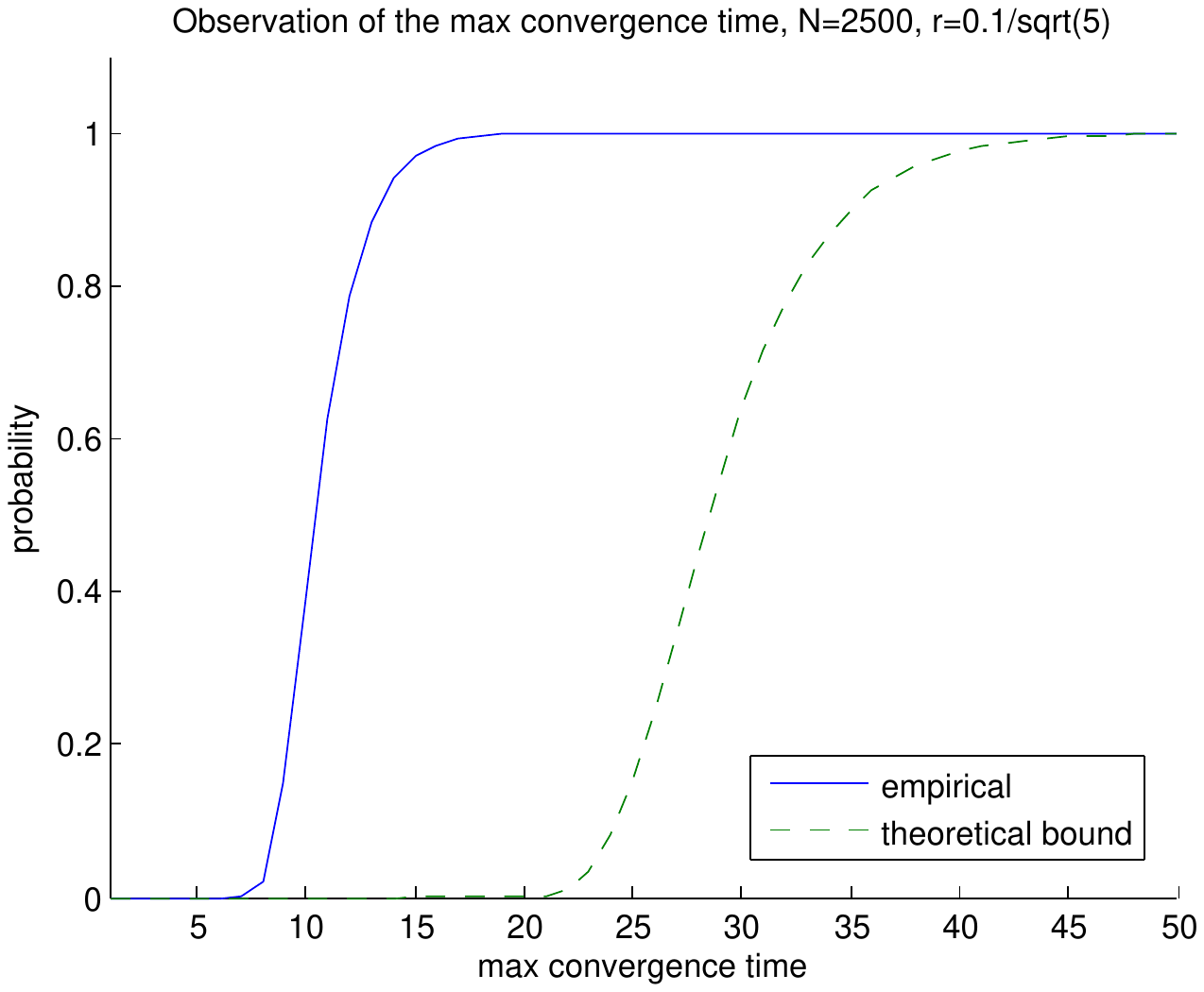}}
\label{fig:boundssubfig2}
}
\end{minipage}
&
\hskip -0.50cm
\begin{minipage}[b]{0.33\linewidth}
\subfigure{
{\includegraphics[scale=0.3, viewport=110 275 470 565]{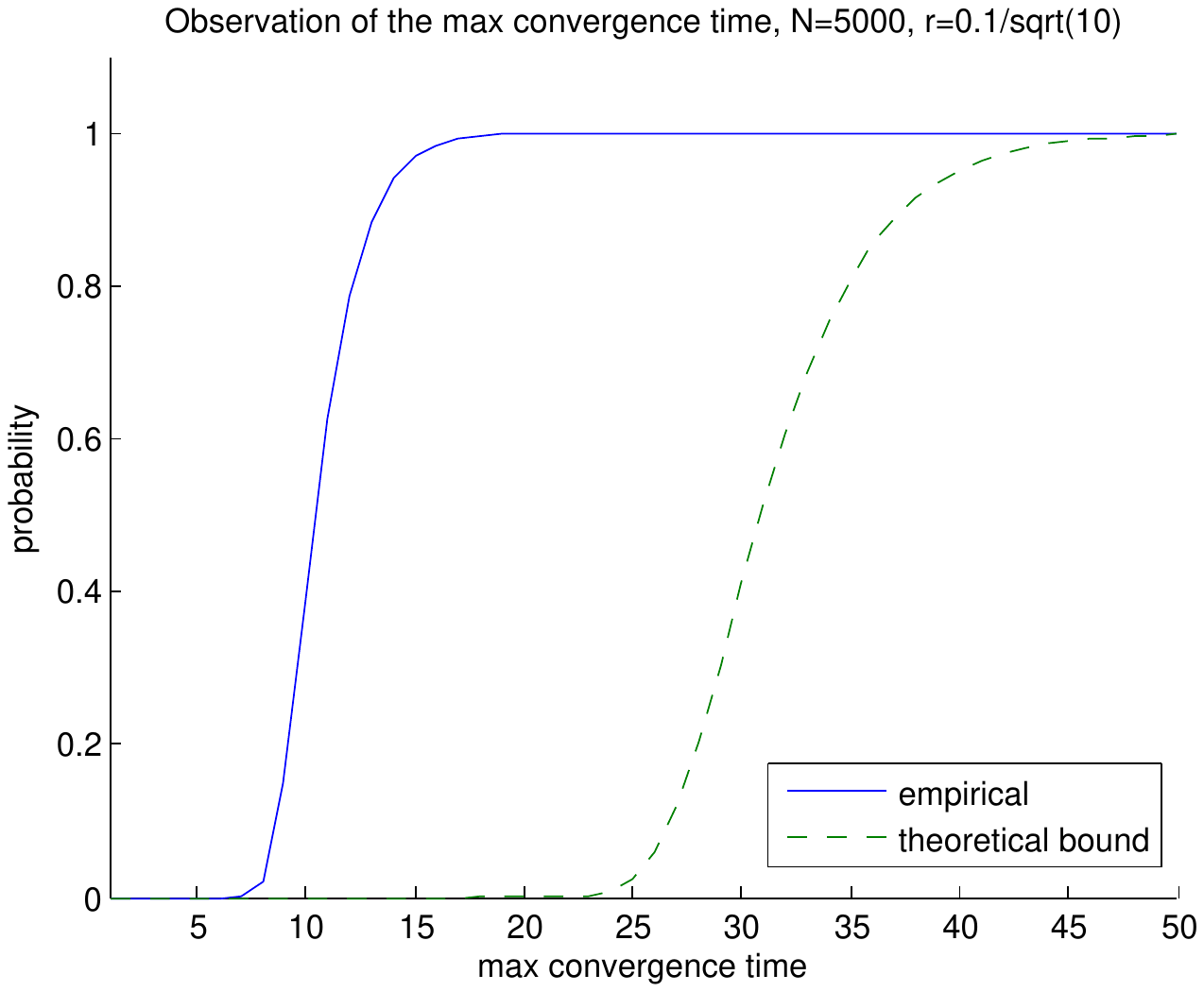}}
\label{fig:boundssubfig3}
}
\end{minipage}\\
\end{tabular}
%
\caption{Numerical validation of Theorem~\ref{th:thbounds}'s bound on the network-wise convergence time. We compare the bound, $P(t_{\max}<k)=(1-(1-q)^k)^N$, with the numerical results, which consider random geometric graphs in which the nodes are randomly placed on the unit square. The charts considers $N \in \{ 500, 2500, 5000 \}$ nodes (from left to right). All experiments considered $2$ listening/signaling periods, interference range of $0.1/\sqrt(\frac{N}{500})$, which result in an average extended degree of $15$, $d_i/T=1$ on average, and $q_i=1/4$.
\label{fig:bounds}}
\end{figure*}

\Section{Convergence within a Neighborhood}
\label{s:markov}
%
%
Theorem~\ref{th:expS} bounds the expected time, ${\cal S}_i$, for a node to reach the relative state \RSTHREE, and follows from Proposition~\ref{prop:2boundq} and equation~$(\ref{eq:Sihiwi})$. Note that ${\cal S}_i \leq 4$ when the number of listening/signaling periods is $n \geq 2$, and considering saturated situations in which the extended node degree $d_i <T$ is smaller than the TDMA frame size. Namely, the proposed algorithm convergence with a neighborhood is brief.

\begin{theorem}[Local Convergence]
\label{th:expS}
The expected time, ${\cal S}_i$, for node $p_i \in P$ to reach the relative state \RSTHREE  satisfies equation~$(\ref{boundht})$, where $n$ is the number of listening/signaling periods, $T$ the TDMA frame size, and $d_i$ is $p_i$'s extended degree.
\B
\begin{equation}
\label{boundht}
\boxed{{\cal S}_i \le \min\{\left(\frac{2n}{n-1}\right)^{\frac{d_i}{T}},\frac{\frac{d_i}{T}+1}{n}\left(\frac{n}{n-1}\right)^{\frac{d_i}{T}+1}\}}
\end{equation}
\end{theorem}

We look into the transition probability among relative states by depicting the diagram of~\Figure~\ref{fig:markov} as an homogeneous Markov chain. We estimate the diagram transition probabilities in a way that maximizes the expected time for reaching the diagram's final state, \RSTHREE.
It is known that the first hitting time is given by ${E_i} \left( T_i^+ \right) = \frac{1}{\pi_i}$, where $\pi = \left( \pi_1, \pi_2, \pi_3 \right)$ is the invariant probability vector~\cite{aldous}. Let ${\cal S}_i$  be the expected time it takes node $p_i$ that starts at the relative state \RSONE to reach \RSTHREE. It is clear that ${\cal S}_i = T_{3}^+ - 1$, because $T_{3}^+ - 1$ is the return time of the relative state \RSTHREE. In our case, the transition matrix $P$ is given by equation~$(\ref{Pxqi})$.
\begin{equation}
\label{Pxqi}
P=\begin{pmatrix} 1-f_i-q_i & f_i & q_i \\ {h_i} & 1-{h_i}-w_i & w_i \\ 1 & 0 & 0\\ \end{pmatrix}
\end{equation}

The invariant probability vector $\pi$ satisfying $\pi P = \pi$ is given by equation~$(\ref{pify})$.

\begin{equation}
\label{pify}
\pi=\frac{\big({h_i}+w_i, f_i, q_i {h_i}+ q_iw_i + f_iw_i\big)}{{h_i}+w_i+ f_i +{h_i} q_i + q_iw_i + f_i w_i}
\end{equation}

The estimation of the maximal expected time necessary to assign the node $p_i$ to a \timeslot requires to compute bounds on the probabilities $f_i$, ${h_i}$, $q_i$ and $w_i$ that maximize equation~$(\ref{hittingtime})$.

\begin{equation}
\label{hittingtime}
{E_{3}}\left(T_{3}^+\right)=\frac{1}{\pi_3}= \frac{{h_i}+w_i+ f_i +{h_i} q_i + q_iw_i + f_i w_i}{q_i {h_i}+ q_iw_i + f_iw_i }
\end{equation}

The expected time for $p_i$ to reach the relative state \RSTHREE is bounded in equation~$(\ref{boundht1})$.

\begin{equation}\label{boundht1}
\boxed{{\cal S}_i={E_{3}}\left(T_{3}^+\right)-1=\frac{{h_i}+w_i+f_i}{q_i{h_i}+q_iw_i+f_iw_i}}
\end{equation}

Equation~$(\ref{boundht})$ has a compact and meaningful bound for equation~$(\ref{boundht1})$. We achieve that by studying the impact of the parameters $T$ and $n$ on the \dstclr. Lemma~\ref{lemma:wq} and equation $(\ref{boundht1})$ imply equation~$(\ref{eq:Sihiwi})$.

\begin{equation}
\label{eq:Sihiwi}
\boxed{{\cal S}_i \le \frac{h_i+w_i+f_i}{q_ih_i+q_iw_i+f_iq_i}=\frac{1}{q_i}}
\end{equation}

\begin{lemma}
\label{lemma:wq}
Suppose that $ n \ge 2$ is the number of listening/signaling periods, see line~\ref{l:max} of the code in \Figure~\ref{f:smp}. Then $w_i\ge q_i$.
\end{lemma}

\begin{proof} Let us consider node $p_i \in P$ that is in relative state \RSONE. Given that $p_i$ has $v_i$ neighbors that compete for the same timeslot, the probability that $p_i$ gets allocated, $q_i\mid_{v_i}$, is given by equation~$(\ref{eq11})$.
\begin{equation}\label{eq11}
q_i\mid_{v_i}=\sum_{k=1}^{n-1} {\rho_k}\left(1-{\rho_1}-\ldots-{\rho_k}\right)^{v_i}
\end{equation}
Consider next that $p_i$ is in relative state \RSTWO, and thus we know that $p_i$ transmitted during the preceding broadcasting round and transited  from relative state \RSONE to \RSTWO. Moreover, $p_i$ is using the same timeslot for the current broadcasting round. The only neighbors of $p_i$ that are using the same timeslot are the neighbors that are also in relative state \RSTWO and, have chosen the same listening/signaling period as $p_i$ during the preceding broadcasting round. Let us denote by $\ell_i$ the number of such neighbors. Given $\ell_i$ the probability $w_i\mid_{\ell_i}$ that $p_i$ is allocated to the timeslot is given by equation~$(\ref{eq22})$.
\begin{equation}\label{eq22}
w_i\mid_{\ell_i}=\sum_{k=1}^{n-1}{\rho_k}\left(1-{\rho_1}-\ldots\ldots-{\rho_k}\right)^{\ell_i}
\end{equation}
We have that $\ell_i$ is stochastically dominated by $v_i$ \cite{ross}, i.e.,  $E(\ell_i)\le E(v_i)$. Indeed, $v_i$ is a random variable that counts the number of neighbors that choose the same timeslot as $p_i$ while $\ell_i$ counts the number of neighbors that choose the same timeslot {\em and} listening/signaling period as $p_i$. For $n\ge 2$, $\ell_i$'s expected value is smaller than $v_i$'s expected value. To conclude, we remark that expressions~$(\ref{eq11})$ and~$(\ref{eq22})$ are the same decreasing function, $f_i \to \sum_ {k=1}^{n-1}{\rho_k}\left(1-{\rho_1}-\ldots -{\rho_k}\right)^{f_i}$, that is evaluated at two different point, $v_i$ and $\ell_i$ respectively. Moreover, since $\ell_i$ is stochastically dominated by $v_i$, equation~($\ref{eq:conslidq}$) holds.
\begin{equation}
\label{eq:conslidq}
w_i=E\left(w_i\mid_{\ell_i}\right)\ge E\left(q_i\mid_{v_i}\right)=q_i
\end{equation}

\end{proof}

Proposition~\ref{prop:2boundq} demonstrates equation~$(\ref{boundq})$ and leads us toward the proof of Theorem~\ref{th:expS}.

\begin{proposition}
\label{prop:2boundq}
Let ${{\rho}_i}=1/MaxRnd$. Equation~$(\ref{boundq})$ bounds from below the probability $q_i$, see Section~\ref{s:compact} of the Appendix.
\end{proposition}
\B
\begin{equation}
\label{boundq}
q_i\ge \max\{\left(\frac{n-1}{2n}\right)^\frac{d_i}{T},\frac{1}{\frac{d_i}{T}+1}\left(1-\frac{1}{n}\right)^{\frac{d_i}{T}+1}\}
\end{equation}

The first bound, $\frac{1}{q_i}\le \left(\frac{2n}{n-1}\right)^{\frac{d_i}{T}}$ (equation~$(\ref{boundht})$), has a simple intuitive interpretation. Let us consider first that two nodes compete for a same timeslot. The two nodes choose independently any of the $n$ listening/signaling periods and there are $n^2$ different possible outcomes. Among these outcomes $n$ correspond to the situation where the two nodes choose the same listening/signaling period and there is no winner. We then have $n^2-n=n(n-1)$ outcomes that lead to a winner. There is then a probability of $n(n-1)/n^2=(n-1)/n$ that one of the node wins the (listening/signaling) competition. Since the game is symmetric, the probability that $p_i$ wins is $(n-1)/(2n)$. The fact that we have $T$ timeslots divides the number of competing nodes, $d_i$, and imply that there are $d_i/T$ competing nodes to the same timeslot. If we interpret the game as a collection of $d_i/T$ independent games, where for each game $p_i$ wins with probability $(n-1)/(2n)$. Thus, the probability $q_i$ that $p_i$ wins is $\left(\frac{n-1}{2n}\right)^\frac{d_i}{T}$. The inverse of this expression gives the average time for the event to occur and is the bound by equation~$(\ref{boundht})$.


\Section{Network Convergence}
\label{s:global}
We estimate the expected time for the entire network to reach a safe configuration in which all nodes are allocated with \timeslots. The estimation is based on the number of nodes that are the earliest to signal in their broadcasting \timeslot. These nodes are winners of the (listening/signaling) competition and are allocated to their chosen \timeslots. However, counting only these nodes leads to under-estimate the number of allocated nodes, which then results in an over-estimation of the convergence time. Indeed, node $p_i \in P$ might have a neighbor $p_j \in {\cal N}_i$ that selects the earliest listening/signaling period in ${\cal N}_i$, but $p_j$ does not transmit because one of its neighbors, $p_k \in {\cal N}_j \setminus {\cal N}_i$, had transmitted in an earlier listening/signaling period. Our bound consider only $p_k$ while both $p_i$ and $p_k$ transmit, became $p_j$ is inhibited by $p_k$'s {beacon}.

Lemma~\ref{th:ewnsigma} shows that the assumption that the nodes are allocated independently of each other's is suitable for bounding the network convergence time, ${\cal S}$. Theorem~\ref{th:thbounds} uses Lemma~\ref{th:ewnsigma} for bounding the network convergence time, ${\cal S}$.

In Section~\ref{s:markov}, we prove a bound on the expected time, ${\cal S}_i$, for a {\em single} node to be allocated to a timeslot. We observe that the bound depends uniquely on the number of listening/signaling periods, $n$, as well as the ratio between the extended degree and the frame size, $d_i/T$. In order to obtain a bound valid for all nodes, we bound this ratio with $x/T$ where $x$ is as defined in Lemma \ref{th:ewnsigma}.
We note that the time needed for the allocation of timeslots to {\em all} the nodes depends on $N$, the total number of nodes.

In detail, the convergence time estimation considers the (fixed and independent) bound, $q_i$, for the probability that a node reach the relative state \RSTHREE within a broadcasting round. Then, the convergence time, $t$, is a random variable with geometric probability, i.e., $P(t$ $=$ $k)$ $=$ $(1$ $-$ $q)^{k-1}q$. Let us denote $t_1, \ldots t_N$ the time it takes for the nodes $p_1,\ldots p_N$ to respectively reach the relative state \RSTHREE. The convergence time, ${\cal S}$, for all the nodes is given by $\max(\{ t_1, \ldots t_N \})$, which depends on $N$. %

\begin{lemma}
\label{th:ewnsigma}
The expected number of nodes, $E(W)$, that win the (listening/signaling) competition after one broadcasting round satisfies equation~$(\ref{eq:ssdkjqwq})$, where $x=\frac{2A}{N}$, $T$ is the number of timeslots, $A$ the number of edges in the interference graph, $G$, and $N= \mid P\mid$ the number of nodes that attempt to access the communication media.
\begin{equation}
\label{eq:ssdkjqwq}
E(W)\ge N\sum_{j=1}^n  {\rho}_j\left(1-\left({\rho}_1+\ldots+{\rho}_j\right)\right)^\frac{x}{T}
\end{equation}
\end{lemma}

\begin{proof}
The nodes that are allocated to a timeslot can previously being on relative state \RSONE or \RSTWO. The probability of a transition from relative state \RSTWO to \RSTHREE is $w_i$, and, a transition from relative state \RSONE to \RSTHREE is $q_i$. As proved in Lemma~\ref{lemma:wq}, we always have $w_i \ge q_i$. To bound the number of nodes that get allocated during a broadcasting round, we use the lower bound on the probability $q_i$ that a node gets allocated to a timeslot. Moreover, in the computations, we use the AM-GM bound~\cite{steele}, which says that if $\sum b_k=1$ then $\prod a_k^{b_k} \le \sum b_k a_k$ and, denote  $d_i$ the number of neighbors of node $p_i$. As proved in Proposition \ref{th:evidit}, since there are $T$ timeslots the number of neighbors of $i$ that choose the same timeslot  as $i$ and compete for it  is bounded by $d_i/T$ . This lemma is proved by equation~$(\ref{eq:sadadahjhkasd})$, where the last line of the expression holds because $\sum_i {d_i}=2 A$.

\begin{eqnarray}
\label{eq:longsadadahjhkasd}
\label{eq:sadadahjhkasd}
E\left(W\right)\ge & \\ E\left(\sum_{i=1}^N 1_{\mid\text{$p_i$ selects the earliest signaling period}}\right)&= \nonumber \\ \nonumber
\sum_{i=1}^N \left({\rho_1}\left(1-{\rho_1}\right)^\frac{d_i}{T}+\ldots {\rho_{n-1}}\left(1-\sum^{n-1}_{k = 1} {\rho}_{k}\right)^\frac{d_i}{T}\right) & = \\\nonumber
\sum_{j=1}^{n}N\sum_{i=1}^N \frac{1}{N}{\rho}_j\left(1- \sum^{j}_{k = 1} {\rho}_{k} \right)^{\frac{d_i}{T}} &\ge \\\nonumber
\end{eqnarray}

\begin{eqnarray}
N\sum_{j=1}^{n} \prod_{i=1}^N {\rho}_j^\frac{1}{N}\left(1 - \sum^{j}_{k = 1} {\rho}_{k} \right)^\frac{d_i}{NT}&=\\\nonumber
N\sum_{j=1}^{n} {\rho}_j\left(1- \sum^{j}_{k = 1} {\rho}_{k} \right)^{\frac{1}{TN}\sum {d_i}}&=\\\nonumber
N\sum_{j=1}^{n}{\rho}_j\left(1-\sum^{j}_{k = 1} {\rho}_{k} \right)^\frac{x}{T}\nonumber
\end{eqnarray}

We note that we use the AM-GM bound to reach the $4$-th row of equation~$(\ref{eq:sadadahjhkasd})$.

\end{proof}

By arguments similar to the ones used in the proof of Proposition~\ref{prop:2boundq}, we deduce that if $N$ nodes compete, the expected number $E(W)$ of nodes that get allocated to a timeslot is lower bounded in equation~$(\ref{eq:sfowmvlkhg})$.
\begin{equation}
\label{eq:sfowmvlkhg}
E(W)\ge N \max\left\{\left(\frac{n-1}{2n}\right)^\frac{x}{T},\frac{\left(\frac{n-1}{n}\right)^{\frac{x}{T}+1}}{\frac{x}{T}+1}\right\}
\end{equation}



Theorem~\ref{th:thbounds} bounds the system convergence time. We numerically validate Theorem~\ref{th:thbounds}, see \Figure~\ref{fig:bounds}. Moreover, our experiments showed that the average convergence time of the network is below the upper bound of equation~$(\ref{boundconv})$.

\begin{theorem}[Global Convergence]
\label{th:thbounds}
%
The expected number of retransmissions is smaller than
$\left(\frac{2n}{n-1}\right)^{d/T}-1$,
where $d= \max(\{ d_i : p_i \in P \})$. Hence, we have that the expected number of broadcasting rounds, ${\cal S}$, that guarantee that all nodes to reach the relative state \RSTHREE satisfies equation~$(\ref{boundconv})$.
\begin{equation}
\label{boundconv}
{\cal S} \le \left(\frac{2n}{n-1}\right)^{d/T}
\end{equation}

Moreover, given that there are $N$ nodes in the network and ${\alpha} \in (0, 1)$, the network convergence time is bounded by equation~$(\ref{boundk})$ with probability $1-{\alpha}$.
\begin{equation}
\label{boundk}
k=1+\frac{log\left(1-\sqrt[N]{1-{\alpha}}\right)}{log\left(1-\left(\frac{n-1}{2n}\right)^\frac{d}{T}\right)}
\end{equation}
This means that with probability ${\alpha}$ all nodes are allocated with timeslots in maximum $k$ broadcasting rounds, see \Figure~$(\ref{fig:boundk})$.
\end{theorem}

\begin{proof}
Theorem~\ref{th:expS} bounds the convergence time of a particular processor, see equation~$(\ref{boundht})$. Lemma~\ref{th:ewnsigma}, see equation ~($\ref{eq:sfowmvlkhg}$) $E(W)\ge N (\frac{n-1}{2n})^{x/T}$, proves that this bound is still valid if we replace the term $d_i/T$ with $x/T$, i.e., we consider the average degree instead of the particular degree of a node. If we replace $x/T$ by $\text{max}\{d_i\}/T$ in expression ~($\ref{eq:sfowmvlkhg}$) we obtain a larger bound because $x/T \le \text{max}\{d_i\}/T$, i.e. $E(W)\ge N (\frac{n-1}{2n})^{x/T}\ge N (\frac{n-1}{2n})^{\text{max}\{d_i\}/T}$.

The bound $E(W)\ge N (\frac{n-1}{2n})^{\text{max}\{d_i\}/T}$ and the discussion in the $1$st paragraph of section \ref{s:global}, show that the number of processors that are allocated during a broadcasting round is bounded by the random variable $\sum_{i=1}^N z_i$, where $z_i$ are identically and independently distributed random variables that are $1$ with probability $(\frac{n-1}{2n})^{\text{max}\{d_i\}/T}$ and $0$ with probability $1-(\frac{n-1}{2n})^{\text{max}\{d_i\}/T}$ (the second random variable dominate the first one, see \cite{lindvall}). This means that we lower bound the number of processors that are allocated if we consider that they are allocated independently with probability $(\frac{n-1}{2n})^{\text{max}\{d_i\}/T}$.

While the processors get allocated to a timeslot, the parameters $d_i$ and $T$ change because some timeslots are no longer available ($T$ decreases and some nodes are allocated $d_i$ decreases). Actually the ratio becomes $\frac{\text{max}\{d_i\} -{h_i}}{T-f_i}$, where ${h_i} \ge f_i$ because if a timeslot is allocated or sensed used by processor $p_i$ then $T$, the number of available timeslots decreases by $1$ and $d_i$, the number of competing nodes, must decrease at least by one since there must be at least one processor that uses the busy timeslot (there may be multiple that are in state \RSTWO). Under these circumstances we always have $\frac{\text{max}\{d_i\} }{T}\ge \frac{\text{max}\{d_i\} -{h_i}}{T- f_i }$. Thus, we can obtain a lower bound for the expected time to reach the relative state \RSTHREE by assuming that all nodes are allocated independently with probability $x=(\frac{n-1}{2n})^{\text{max}\{d_i\}/T}$. We simplify the following arguments by using this define of $x$.

To bound the number of broadcasting rounds we consider the following game. The bank pays $1$ unit to the nodes that get in state \RSTHREE (get allocated to a timeslot), and receives  $x/(1-x)$ units per nodes that fails to get in state \RSTHREE. The game is fair because in each round the expected gain is $1\times x - x/(1-x)\times (1-x)=0$. If we denote by $W_i$ the number of processors that get in state \RSTHREE during the $i$-th broadcasting round and by $L_i$ the number of processors that fail we have that the gain is given by equation~\ref{eq:interdeb}, where $t$ denotes the total number of rounds.
\begin{equation}\label{eq:interdeb}
\text{gain}=\sum_{i=1}^t\bigl(\frac{x}{1-x} L_i-W_i\bigr)
\end{equation}
The expected gain is $0$ because the game is fair ($E(gain)=0$) and $\sum_{i=1}^t W_i=N$ because eventually all the nodes get in state \RSTHREE and the bank pays $1$ unit for each such processors. If we compute the expectation on both sides of equation~$(\ref{eq:interdeb})$, we then obtain equation~\ref{eq:sprjkbd}.
\begin{equation}\label{eq:sprjkbd}
N=\frac{x}{1-x} E\bigl(\sum_{i=1}^t L_i\bigr)
\end{equation}
We observe that $E(\sum_{i=1}^t L_i)$ is the expected total number of retransmissions and $E(\sum_{i=1}^t L_i)/N$ is the average expected number of retransmissions whose value is $(1-x)/x$. replacing $x$ with its expression, we obtain that the average number of retransmission is bounded by $(2n/(n-1))^{\text{max}\{d_i\}/T}-1$ and, this leads to the bound equation~$(\ref{boundconv})$.

To prove the second assertion, let $t_1,\ldots, t_N$ be the convergence time of nodes $1, \ldots N$, respectively. The random variables, $t_i$, are bound by random variables with geometric random distribution with expectation of $(2n/(n-1))^{d/T}$, with $d=max\{ d_i : d_i\in P\}$. We require that $t_{\max}=\max\{t_1,\ldots,t_N\}$ in order to ensure that all nodes are allocated with timeslots. The fact that the random variables, $t_i$, are independent and identically distributed, implies equation~$(\ref{eq:rviid})$, where $t$ is a random geometrical random variable, i.e., $\Pr(t={k{^\prime}})=(1-q)^{{k{^\prime}}-1} q$ and $\Pr(t \ge {k{^\prime}})=(1-q)^{{k{^\prime}}-1}$.

\begin{eqnarray}
\label{eq:rviid}
\Pr\left(t_{\max}\le {k{^\prime}}\right)=P\left(t_1\le {k{^\prime}},\ldots, t_N \le {k{^\prime}} \right)=\\\nonumber
\Pr\left(t_1\le {k{^\prime}} \right)\cdot \ldots \cdot P\left(t_N\le {k{^\prime}} \right)=P\left(t \le {k{^\prime}} \right)^N
\end{eqnarray}
Which $t_{\max}\le {k{^\prime}}$ satisfies equation~$(\ref{eq:rvdloe})$ with probability ${\alpha}$?
\begin{eqnarray}
\label{eq:rvdloe}
\Pr(t_{\max}< {k{^\prime}})=\Pr\left(t< {k{^\prime}}\right)^N=\\\nonumber
\left(1-\left(1-q\right)^{{k{^\prime}}-1}\right)^N\ge 1-{\alpha}
\end{eqnarray}
By solving equation~$(\ref{eq:rvdloe})$, we observe that equation~$(\ref{eq:rvdloe})$ is satisfied for any ${k{^\prime}} \geq k$, where $k$ satisfies equation~$(\ref{boundk})$. This proves that, with probability $1-{\alpha}$, the network convergence time is bounded by equation~$(\ref{boundk})$.
\end{proof}

\Section{Implementation}
\label{s:imp}
Existing MAC protocols offer mechanisms for dealing with contention (timeslot exhaustion) via policies for administering message priority, such as~\cite{DBLP:journals/cn/RomT81}. In particular, the IEEE {\tt 802.11p} standard considers four priorities and techniques for facilitating their policy implementation. We explain similar techniques that can facilitate the needed mechanisms.

\Subsection{Prioritized listening/signaling periods}
We partition the sequence of listening periods, $[0, MaxRnd)$, into $MaxPrt$ subsequences,  $[0, MaxRnd_{0}), \ldots [MaxRnd_{MaxPrt-2}, MaxRnd_{MaxPrt-1})$, where $[MaxRnd_{k-1}, MaxRnd_{k})$ is used only for the $k$-th priority. E.g., suppose that there are six listening/signaling periods, and that nodes with the highest priority may use the first three listening/signaling periods, $[0, 2]$, and nodes with the lowest priority may use the last three, $[3, 5]$. In the case of two neighbors with different listening period parameters, say $[0, 2]$ and $[3, 5]$, that attempt to acquire the same broadcasting timeslot, the highest priority node always attempts to broadcast before the lowest priority one.

\Subsection{TDMA-based back-off}
Let us consider two back-off parameters, $CW_{start}$ and $CW_{end}$, that refer to the maximal and minimal values of the contention window. Before selecting an unused timeslot, the procedure counts a random number of unused ones. \Figure~\ref{f:slct} presents an implementation of the \op{select$\_$unused}$()$ function that facilitates back-off strategies as an alternative to the implementation presented in line~\ref{l:uniformselect} of \Figure~\ref{f:smp}.

\begin{figure}[t]
\begin{flushleft}
\begin{minipage}[b]{0.99\linewidth}
\begin{center}
\fbox{\hspace*{+10pt}{
\begin{lstlisting}[language=ioaNums]
*Addtional constants and variables*
$~~$$CW_{start}$ and $CW_{end}$ : backoff parameters
$~~$$count$ : statically allocated variable that counts the backoff steps

*Function* !select_unused!(set) 
$~~$let $rtn\_val = \bot$v $//$ indicate busy channel (default $return$ value) $\label{l:default}$
$~~$if $count \leq 0$ then $count \gets uniform\_select([CW_{start}, CW_{end}])$ $\label{l:zero}$
$~~$$count \gets count - \mid set \mid$ $\label{l:wtrtrn}$
$~~$if $count \le 0$ then $(count, rtn\_val) \gets (0, uniform\_select(set))$ $\label{l:smpl}$
$~~$return $rtn\_val$ $\label{l:rtnbot}$
\end{lstlisting}}}
\end{center}
\end{minipage}
\end{flushleft}
\caption[]{\op{select$\_$unused}$()$ with TDMA-based back-off}\label{f:slct}
\end{figure}

The statically allocated variable $count$ records the number of backoff steps that node $p_i$ takes until it reaches the zero value. Whenever the function \op{select$\_$unused}$()$ is invoked with $count_i=0$, node $p_i$ assigns to $count_i$  a random integer from $[CW_{start}, CW_{end}]$ (cf. line~\ref{l:zero}). Whenever the value of $count_i$ is not greater than the number of unused timeslots, the returned timeslot is selected uniformly at random (cf. lines~\ref{l:wtrtrn} to~\ref{l:smpl}). Otherwise, a $\bot$-value is returned after deducting the number of unused timeslots during the previous broadcasting round (cf. lines~\ref{l:default} and~\ref{l:rtnbot}).


\begin{figure}[t]
\begin{center}
\begin{minipage}{0.9\textwidth}
{\includegraphics[scale=0.80,viewport=100 270 470 572]{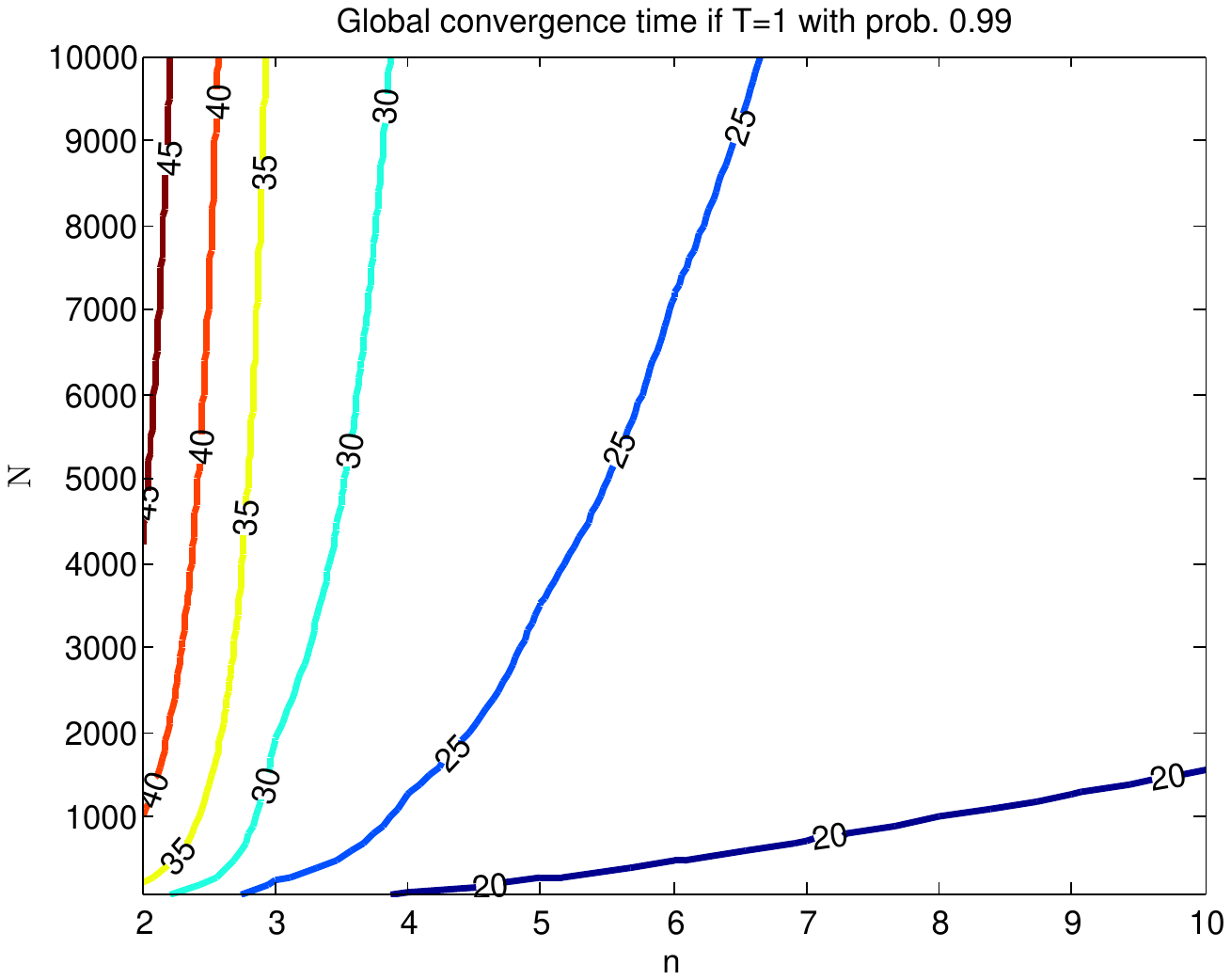}}
\end{minipage}
\end{center}
\begin{center}
\caption{Contour plot of equation~$(\ref{boundk})$ for $s=d/T=1$. {\em Contour charts}~\cite{crocker1947frontal} present two parameter functions, e.g., the convergence time function, $k(n,N)$ presented in equation~$(\ref{boundk})$.
Contour lines in~\Figure~\ref{fig:boundk} connect values of $k(n,N)$ that are the same (see the text tags along the line). When $N$ nodes attempt to access the medium, the conversance time, ${\cal S}$ (cf. the contour lines), is stable in the presence of a growing number, $n$, of listening/signaling periods.}\label{fig:boundk}
\end{center}

\BB\B\B
\end{figure}

\Section{Discussion}
\label{s:dis}

Thus far, both schedule-based and non-schedule-based MAC algorithms could not consider timing requirements within a provably short recovery period that follows (arbitrary) transient faults and network topology changes. This work proposes the first self-stabilizing TDMA algorithm for \manet that has a provably short convergence period. Thus, the proposed algorithm possesses a greater predictability degree, whilst maintaining low communication delays and high throughput.

In this discussion, we would like to point out the algorithm's ability to facilitate the satisfaction of severe timing requirements for \manet by numerically validating Theorem~\ref{th:thbounds}.
As a case study, we show that, for the considered settings of \Figure~\ref{fig:bounds}, the global convergence time is brief and definitive.
\Figure~\ref{fig:boundk} shows that when allowing merely a small fraction of the bandwidth to be spent on frame control information, say three listening/signaling periods, and when considering $99\%$ probability to convergence within a couple of dozen TDMA frames, the proposed algorithm demonstrates a low dependency degree on the number of nodes in the network even when considering $10,000$ nodes.
We have implemented the proposed algorithm, extensively validated our analysis via computer simulation, and tested it on a platform with more than two dozen nodes~\cite{MPSTT12}. These results indeed validate that the proposed algorithm can indeed facilitate the implementation of MAC protocols that guarantee satisfying these severe timing requirements.

The costs associated with predictable communications, say, using cellular base-stations, motivate the adoption of new networking technologies, such as MANETs and VANETs. In the context of these technologies, we expect that the proposed algorithm will contribute to the development of MAC protocols with a higher predictability degree.

\smallskip
\noindent{\bf Acknowledgments}
We thank Thomas Petig for improving the presentation.


\bibliographystyle{plain}
\bibliography{rwork}
\pagebreak~\\
\noindent {\Large {\bf Appendix}}\\

\Section{Properties~$(\ref{eq:invarIa})$ to~$(\ref{eq:invarIIai})$}
\label{s:invarI}
Propositions~\ref{l:invarI},~\ref{l:invarI1} and~\ref{l:invarIIb} imply that properties~$(\ref{eq:invarIa})$,~$(\ref{eq:invarIb})$, and respectively,~$(\ref{eq:invarIIai})$ hold within two broadcasting rounds (or one complete broadcasting round). Let $R$ be an execution of the \dstclr, $x>0$ an integer, and  $c{^{_{\op{timeslot}}}_{^{0}}}(x)$ the first configuration in a complete broadcasting round $R(x)=(c{^{_{\op{timeslot}}}_{^{0}}}(x), \ldots$ $c{^{_{\op{carrier\_sense/receive}}}_{^{T-1}}}(x))$. We note that  $c{^{_{\op{timeslot}}}_{^{0}}}(x)$ follows an arbitrary starting configuration.

Proposition~\ref{l:invarI} shows that, within a broadcasting round from $c{^{_{\op{timeslot}}}_{^{0}}}(x)$, Property~$(\ref{eq:invarIa})$ holds.

\begin{proposition}
\label{l:invarI}
In $c{^{_{\op{\timeslot}}}_{^{0}}}(x+1)$, it holds that $signal_i=\op{false}$.
\end{proposition}
\begin{proof}
The value of $signal_i$ is updated in line~\ref{l:for} (assigned to \op{true}) and in lines~\ref{l:outdated}, \ref{l:quit}, and~\ref{l:emptr} (assigned to \op{false}). Let us look into these assignments.

In every \timeslot, the value \op{false} is assigned to $signal_i$ (cf. line~\ref{l:outdated}). Suppose that the function \op{send}$()$ is called, and thus, \op{true} is assigned to $signal_i$ (line~\ref{l:for}). We proposition that before returning from the function \op{send}$()$ and after \op{true} is assigned to $signal_i$ (line~\ref{l:for}), node $p_i$ must assign  \op{false} to $signal_i$ either in line~\ref{l:quit} or~\ref{l:emptr}. To see that, let us look at lines~\ref{l:for} and~\ref{l:p}. Eventually either $signal_i = \op{false}$ (because of an assignment in line~\ref{l:emptr}) or ${\rho(k)}=\op{true}$ (line~\ref{l:p}) holds (note the condition when $k=MaxRnd$). The latter case implies the execution of line~\ref{l:quit}.
\end{proof}

Proposition~\ref{l:invarI1} shows that, within a broadcasting round from $c{^{_{\op{timeslot}}}_{^{0}}}(x)$, Property~$(\ref{eq:invarIb})$ holds.

\begin{proposition}
\label{l:invarI1}
$(\exists t \in unused\_set_i \setminus \{ s_i \}) \leftrightarrow (\nexists p_k \in {\cal N}_i : s_k = t)$ in $c{^{_{\op{timeslot}}}_{^{0}}}(x+1)$.
\end{proposition}
\begin{proof}
Recall that $unused\_set_i = \{ k : unused_i[k] = \op{true} \}$ (see line~\ref{l:emptyset}) and that the proposition statement does not consider the cases in which: (1) $s_i = s_k$ (because $t \neq s_i$) in $c{^{_{\op{timeslot}}}_{^{0}}}(x+1)$, or (2) There exists a configuration $c \in R(x)$, such that $s_k \neq \bot$ in $c$ and $s_k = \bot$ in $c{^{_{\op{timeslot}}}_{^{0}}}(x+1)$ (because by $unused$\_$set$'s definition, $\bot$ is never in $unused\_set_i$).

We note that in every broadcasting round, node $p_k \in P$ {\em at most once}: (1) Allocates the broadcasting \timeslot $s_k$ (when $t_k = 0$, see line~\ref{l:slcemp}), (2) Transmits a packet (when $t_k = s_k$, see line~\ref{l:fetch}), and (3) Deallocates the broadcasting \timeslot $s_k$ (by assigning $\bot$ to $s_k$ when $t_k = s_k$ and the $\op{carrier$\_$sense}(t)$ event is raised, see line~\ref{l:state}). Moreover, node $p_i$ updates $unused_i[t]$ only in lines~\ref{l:outdated} (\op{true}) and~\ref{l:emptr} (\op{false}), when $p_i$ removes stale information just before \timeslot $t$, and respectively, when the event $\op{carrier$\_$sense}(t)$ is raised.

Line~\ref{l:outdated} is executed at the start of every \timeslot, whereas line~\ref{l:emptr} is executed after, and only when the event $\op{carrier$\_$sense}(t)$ is raised. The event $\op{carrier$\_$sense}(t)$ is raised after, and only when the node $p_k \in {\cal N}_i$ transmits in \timeslot $t$. In other words, {\em none} of $p_i$'s neighbors, $p_k \in {\cal N}_i$, that transmits in \timeslot $s_k = t$, can avoid causing the event $\op{carrier$\_$sense}(t)$ to be raised, and \timeslot $t$ to be included in $unused\_set_i \setminus \{ s_i \}$.
\end{proof}


%
%

Proposition~\ref{l:invarIIb} shows that, within a broadcasting round from $c{^{_{\op{timeslot}}}_{^{0}}}(x)$, Property~$(\ref{eq:invarIIai})$ holds.

\begin{proposition}
\label{l:invarIIb}
$(s_i \neq \bot) \vee (unused\_set_i \setminus \{ s_i \} \neq \emptyset)$ holds in $c{^{_{\op{timeslot}}}_{^{0}}}(x+1)$.
\end{proposition}
\begin{proof}
If $s_i \neq \bot$ in $c{^{_{\op{\timeslot}}}_{^{0}}}(x+1)$, we are done. Let us suppose that $s_i = \bot$ in $c{^{_{\op{\timeslot}}}_{^{0}}}(x+1)$ and show that $unused\_set_i \setminus \{ s_i \} \neq \emptyset$ in $c{^{_{\op{\timeslot}}}_{^{0}}}(x+1)$.

Let us assume, in the way of proof by contradiction that, $unused\_set_i \setminus \{ s_i \} = \emptyset$ and show that ${d_i}/{T} > 1$, i.e., a contradiction with the assumption that $\forall p_i \in P : {d_i}/{T} \lneq 1$.

Recall that $unused\_set_i = \{ k : unused_i[k] = \op{true} \} \subseteq [0, T-1]$ (see line~\ref{l:emptyset}). Therefore, the assumption that $s_i = \bot$ implies that $unused\_set_i = unused\_set_i  \setminus \{ s_i \} \subseteq [0, T-1]$, because by $unused$\_$set$'s definition, $\bot$ is never in $unused\_set_i$.

By Proposition~\ref{l:invarI1}, we can say that $\forall t \in [0, T-1] : (\nexists t \in unused\_set_i) \leftrightarrow (\exists p_k \in {\cal N}_i : s_k = t)$. Since $unused\_set_i \subseteq [0, T-1]$, we can write $[0, T-1] \setminus unused\_set_i \subseteq \{ s_k \in [0, T-1] : p_k \in {\cal N}_i \}$. By the fact that $unused\_set_i = \emptyset$, we have that $T \leq \mid \{ s_k \in [0, T-1] : p_k \in {\cal N}_i \} \mid$. Since $d_i = \mid {\cal N}_i \mid$ (by definition), we have that $\mid \{ s_k \in [0, T-1] : p_k \in {\cal N}_i \} \mid \leq d_i$, which implies $T \leq d_i$: a contradiction with the assumption that ${d_i}/{T} \lneq 1$.
\end{proof}

\Section{Properties~$(\ref{eq:invarIIa})$ to~$(\ref{eq:invarIIb})$}
\label{s:invarII}
Section~\ref{s:invarI} of this Appendix shows that, starting from an arbitrary configuration, node $p_i \in P$ enters the relative state \RSONE within two broadcasting rounds. This section shows considers the probability for $p_i$ to enter the relative states \RSTWO and \RSTHREE.

Let $x>0$ and $R$ be an execution of the \dstclr. Suppose that $c{^{_{\op{timeslot}}}_{^{0}}}(x)$ is the first configuration in a complete broadcasting round $R(x)$ for which properties~$(\ref{eq:invarIa})$ to~$(\ref{eq:invarIIai})$ hold in configuration $c{^{_{\op{timeslot}}}_{^{0}}}(x)$ with respect to node $p_i \in P$, i.e., $p_i$ is in relative state \RSONE, \RSTWO or \RSTHREE. Propositions~\ref{l:invarI2},~\ref{l:invarII1} and~\ref{l:invarII2} show that there is a nonzero probability that node $p_i$ enters the relative state \RSTHREE from either \RSONE or \RSTWO in configuration $c{^{_{\op{timeslot}}}_{^{0}}}(x+1)$.

Proposition~\ref{l:invarI2} shows that $p_i$ attempts to broadcast once in every round.

\begin{proposition}
\label{l:invarI2}
During broadcasting round $R(x)$, $p_i$ executes line~\ref{l:fetch} and calls the function $\op{send}()$.
\end{proposition}
\begin{proof}
If $s_i \neq \bot$ in $c{^{_{\op{\timeslot}}}_{^{0}}}(x)$, we are done by lines~\ref{l:slcemp} and~\ref{l:fetch}.
Let us consider the case of $s_i = \bot$ in $c{^{_{\op{\timeslot}}}_{^{0}}}(x)$. By Property~$(\ref{eq:invarIIa})$, $unused\_set_i \neq  \emptyset$ and thus when line~\ref{l:slcemp} is executed, the function $\op{select$\_$unused}()$ returns a non-$\bot$ element from $unused\_set_i$ and $s_i \neq \bot$ when executing line~\ref{l:fetch}.
\end{proof}

Propositions~\ref{l:invarII1} and~\ref{l:invarII2} consider the set $M_i(x+1) = \{ p_k \in {\cal N}_i : s_k = t^{\prime} \}$ and the number $m_i = | M_i(x+1) |$ of $p_i$'s neighbors that attempt to broadcast during $p_i$'s \timeslot, $t^{\prime}$, of broadcasting round $R(x)$.

Let ${\rho}_j$ be the probability for $p_i$ to transmit in the $j$-th listening/signaling period of \timeslot ${t^{\prime}}$ (cf. line~\ref{l:p}). This paper considers the concrete transmission probability $\rho_i = 1/MaxRnd$. We motivate our implementation choice of the transmission probability, ${{\rho}_i}$, in \Figure~\ref{f:rho}. Note that the {\it sequential} selection of the broadcasting rounds with probability $1/(MaxRnd-k+1)$ leads to the uniform selection $\rho_k=1/MaxRnd$.

\begin{figure*}[t!]
\begin{center}
\fbox{
\begin{minipage}[b]{0.9\linewidth}
Defining optimal transmission probabilities for any choices of $T, n, d_i$ is not possible. We choose to consider and look for optimal choices when $d_i \simeq T$ (the 'hard' case) and make a case for a uniform probability ${{\rho}_i} = \frac{1}{n} : i \in [1, n]$.

Let us consider node $p_i \in P$ that competes, together with $k-1$ other neighbors, for the same unique \timeslot. The probability that node $p_i$ wins the (listening/signaling) competition is ${\rho}_1 (1-{\rho}_1)^{k-1}$, where ${\rho}_1$ is the probability of choosing the first listening/signaling period. The value ${\rho}_1=\frac{1}{k}$ maximizes this probability.
In the more general case where there is more than one \timeslot, we consider a strategy that aims at guessing the number, $k$, of competing neighbors, which the optimal probability of transmission depends on. During the first listening/signaling period, the strategy considers the case in which there are $n=MaxRnd$ signaling nodes, and thus, the transmission probability is $1/MaxRnd$, where $MaxRnd \simeq T$. During the second listening/signaling period, the strategy considers the case in which there are $MaxRnd-1$ neighbors, and thus, the transmission probability is $1/(MaxRnd-1)$, and so on. This {\it sequential} selection of the listening/signaling period leads to a uniform choice of a listening/signaling neighbor.
The above strategy is driven by a heuristic in which nodes signal with probability that is optimal for the case of $n \simeq T$, and thus, it depends on the number of competing neighbors.

\end{minipage}
}
\end{center}
\begin{center}
\BB\B
\caption{Transition probability, ${{\rho}_i}$, for listening/signaling periods (line~\ref{l:p} in \Figure~\ref{f:smp})
\label{f:rho}}
\end{center}
%
%
\BBB
\end{figure*}

Proposition~\ref{l:invarII1} considers $p_i$'s chances to be the only one to transmit in its neighborhood.

\begin{proposition}
\label{l:invarII1}
There is a nonzero probability, $OnlyOne_i(x)$ (cf. equation~$(\ref{eq:cfposur})$), that only node $p_i$ transmits in its broadcasting \timeslot, ${t^{\prime}}$, of broadcasting round $R(x)$.
\begin{normalfont}
\begin{eqnarray}
&OnlyOne_i(x)\mid_{m_i>0}={\rho_1}(1-{\rho_1})^{m_i}+{\rho_2}(1-{\rho_1}-{\rho_2})^{m_i} \nonumber \\
&+ \ldots +{\rho_{n-1}}(1-\sum^{n-1}_{\ell=1} {\rho}_k)^{m_i}
\label{eq:cfposur}
\end{eqnarray}
\end{normalfont}
\end{proposition}
\begin{proof}
We show that there is a nonzero probability that only node $p_i$ transmits in its broadcasting \timeslot, ${t^{\prime}}$, of broadcasting round $R(x)$. Let us look at $p_i$ and the nodes in $M_i(x)$ while they attempt to broadcast in the steps $a_i^{{_{\op{timeslot},}}_{t^{\prime}}}(x)$ and $a_k^{{_{\op{timeslot},}}_{t^{\prime}}}(x)_{|k \in M_i(x)}$. All of these steps include the execution of line~\ref{l:p}, viz., each node chooses to transmit in listening/signaling period $\ell \in [0, MaxRnd]$ with probability ${\rho_\ell}=1/(MaxRnd-\ell)$. Therefore, for any $MaxRnd > 0$, there is a nonzero probability, $OnlyOne_i(x)$ that, during \timeslot $t^{\prime}$, node $p_i$ transmits in the listening/signaling period $a \in MaxRnd$ and no node in $M_i(x)$ transmits in round $a$ (or in an earlier one).

We note that the fact that $p_i$ transmits first during \timeslot $t^{\prime}$ implies that it is the only to transmit during $t^{\prime}$. This is because once $p_i$ transmits a {beacon} in step $a_i^{{_{\op{timeslot},}}_{t^{\prime}}}(x)$ (which includes the execution of line~\ref{l:starta}) node $p_j \in {\cal N}_i \supseteq M_i(x)$ raises the event $\op{carrier$\_$sense}(t^{\prime})$ immediately after $a_i^{{_{\op{timeslot},}}_{t^{\prime}}}(x)$. Thus, $\forall p_j \in M_i(x)$ we have that immediately after step $a_i^{{_{\op{timeslot},}}_{t^{\prime}}}(x)$, node $p_j$ takes step $a_j^{{_{\op{carrier$\_$sense},}}_{^{t^{\prime}}}}(x)$, which includes the execution of lines~\ref{l:state} and~\ref{l:emptr} that assigns $\bot$ to $s_j$ and $false$ to $signal_j$. Thus, $p_j$ leaves the (listening/signaling) competition for \timeslot $t^{\prime}$ (see line~\ref{l:for}) and does not transmits its \op{DATA} packet (see line~\ref{l:transmit}).

We now turn to calculate $OnlyOne_i(x)$. Let the variable ${m_i} = \mid M_i(x) \mid$ denote the number of nodes that select the same \timeslot as $p_i$ in configuration $c{^{_{\op{\timeslot: s$\neq$$\bot$}}}_{^{0}}}(x)$. The value of $OnlyOne_i(x)$ depends on the value of $m_i$ and we denote this dependence with the notation $q(i)\mid_{m_i}$ (conditional probability). It means the value of $OnlyOne_i(x)$ depends on the value of ${m_i}$. The value of $OnlyOne_i(x)$ for ${m_i}=0$ is $OnlyOne_i(x)\mid_{{m_i}=0}=1$. For the case of ${m_i} > 0$, $OnlyOne_i(x)$'s value is given by equation~$(\ref{eq:cfposur})$) (that appears again below), where ${\rho}_j$ is the probability for transmitting in the $j$-th listening/signaling period.

\begin{normalfont}
\begin{eqnarray}
&OnlyOne_i(x)\mid_{m_i>0}={\rho_1}(1-{\rho_1})^{m_i}+{\rho_2}(1-{\rho_1}-{\rho_2})^{m_i} \nonumber \\
&+ \ldots +{\rho_{n-1}}(1-\sum^{n-1}_{\ell=1} {\rho}_k)^{m_i} ~~[\mathrm{clone~of~equation}~(\ref{eq:cfposur})] \nonumber
\end{eqnarray}
\end{normalfont}

We note that the $j$-th term in equation~$(\ref{eq:cfposur})$, is the probability that node $p_i$ selects the $j$-th listening/signaling period and all its neighbors select a later listening/signaling period.
\end{proof}

Proposition~\ref{l:invarII2} shows that once a node is the only one in its neighborhood to transmit during its broadcasting \timeslot, it enters the relative state \RSTHREE.

\begin{proposition}
\label{l:invarII2}
$M_i(x) = \emptyset$ (or having none of the nodes in $M_i(x)$ transmitting during \timeslot ${t^{\prime}}$)
implies that node $p_i$ is in the relative state \RSTHREE in $c{^{_{\op{timeslot}}}_{^{0}}}(x+1)$.
\end{proposition}

\begin{proof}
We need to show that, in $c{^{_{\op{timeslot}}}_{^{0}}}(x+1)$, we have that $s_i = t^{\prime} \neq \bot$ and $\forall p_j \in {\cal N}_i : s_i \neq s_j$.

\noindent {\bf Showing that $s_i = t^{\prime} \neq \bot$ in $c{^{_{\op{timeslot}}}_{^{0}}}(x+1)$~~~~~~} The proposition assumes that $t^{\prime} \neq \bot$ in $c{^{_{\op{timeslot}}}_{^{0}}}(x)$. We wish to show that $s_i = t^{\prime}$ in $c{^{_{\op{timeslot}}}_{^{0}}}(x+1)$, which implies that $s_i \neq \bot$ holds in $c{^{_{\op{timeslot}}}_{^{0}}}(x+1)$ and throughout $R(x+1)$.

Since the variable $s_i$ is assigned only in lines~\ref{l:slcemp} (when $t_i =0$) and~\ref{l:state} (when $t_i = t^{\prime}$), it is sufficient to show that line~\ref{l:state} is not executed by any step during \timeslot $t^{\prime}$ of broadcasting round $R(x)$, i.e., $a_i^{{_{\op{carrier$\_$sense},}}_{^{t^{\prime}}}}(x) \not \in R(x)$.

Node $p_i$ raises the event \op{carrier$\_$sense} only during \timeslots in which $p_i$'s neighbor, $p_j$, transmits. By the proposition assumptions that, during \timeslot $t^{\prime}$ of broadcasting round $R(x)$, none of $p_i$'s neighbors transmits, we have $a_i^{{_{\op{carrier$\_$sense},}}_{^{t^{\prime}}}}(x) \not \in R(x)$. Moreover, $a_i^{{_{\op{timeslot},}}_{t^{\prime}}}(x+1)$ does not includes an execution of line~\ref{l:slcemp} that changes the value of $s_i$, because $s_i = t^{\prime} \neq \bot$ in $c{^{_{\op{timeslot}}}_{^{0}}}(x+1)$.

\noindent {\bf Showing that $\forall p_j \in {\cal N}_i : s_i \neq s_j$ in $c{^{_{\op{timeslot}}}_{^{0}}}(x+1)$~~~~~~} The proposition assumes that $\forall p_j \in {\cal N}_i : s_i \neq s_j$ in $c{^{_{\op{timeslot}}}_{^{0}}}(x)$. We wish to show that the same holds in $c{^{_{\op{timeslot}}}_{^{0}}}(x+1)$. Since the variable $s_j$ is assigned to a non-$\bot$ value only in line~\ref{l:slcemp} when $t_i =0$, it is sufficient to show that when line~\ref{l:slcemp} is executed in step $a_j^{{_{\op{timeslot},}}_{0}}(x+1)$ the function $\op{select$\_$unused}()$ considers a set that does not includes $p_i$'s \timeslot, $s_i$. This is implied by the facts that $\forall p_j \in {\cal N}_i : unused_j[t^{\prime}] = \op{false}$ (Claim~\ref{l:invarI2}) and $s_i = t^{\prime}$ (first item of (II) of this proof) in $c{^{_{\op{timeslot}}}_{^{0}}}(x+1)$.
\end{proof}

\Section{Bounding $OnlyOne_i(x)$}
\label{s:compact}
Propositions~\ref{prop:2boundq} and~\ref{th:ndoslo} bound $OnlyOne_i(x)$'s value, where $R(x)$ is the $x$-th broadcasting round in execution $R$ of the \dstclr. We assume that properties~$(\ref{eq:invarIa})$ to~$(\ref{eq:invarIIb})$ holds in the first configuration, $c{^{_{\op{\timeslot}}}_{^{0}}}(x)$, of $R(x)$.
These bounds are obtained by computing the expectation of $q_i \mid_{m_i}$ with respect to ${m_i}$, where $M_i(x) = \{ p_k \in {\cal N}_i : s_k = t^{\prime} \}$ in $c{^{_{\op{\timeslot}}}_{^{0}}}(x)$ and $m_i = | M_i(x) |$. The reason is that ${m_i}$ is a random variable, i.e., $q_i=E\left(OnlyOne_i({x})\mid_{{m_i}}\right)$, where the expectation is computed with respect to the random variable ${m_i}$.

We note that all the terms in equation~$(\ref{eq:cfposur})$ are convex functions of ${m_i}$. This means that by Jensen's inequality, we obtain a lower bound of $q_i$ in equation~$(\ref{eq:jensen})$ by evaluating the expression $q_i \mid_{m_i}$ at ${m_i}$'s expectation, $E({m_i})$~\cite{jensen1906fonctions}.

\begin{eqnarray}
\label{eq:jensen}
q_i&=&E \left( q_i \mid_{m_i} \right) \ge q_i \mid_{E({m_i})}
\end{eqnarray}

The expression on the right side of the inequality can be again lower bounded if we estimate an upper bound for $E(m_i)$. We proceed to the computations in the proof of the Proposition~\ref{th:ndoslo} after demonstrating Proposition~\ref{th:evidit} which shows that $E({m_i})$ is bounded by the ratio $d_i/T$, which is rather intuitive but, needs to be proved.

\begin{proposition}
\label{th:evidit}
In configuration $c{^{_{\op{\timeslot}}}_{^{0}}}(x)$ in holds that $E({m_i}) \le d_i / T$, where $m_i = | M_i(x) |$.
\end{proposition}

\begin{proof}
We show that $E({m_i})=d_i/T$ by considering configuration $c{^{_{\op{timeslot}}}_{^{0}}}(x)$.
%
%
%
%
The maximal number of $p_i$'s neighbors that might choose the same \timeslot as $p_i$ in configuration $c{^{_{\op{timeslot}}}_{^{0}}}(x)$ is $\sum_{{p_j \in {\cal N}_i}}1_{\{s_j=\bot\}}$, because any node, $p_j \in {\cal N}_i$, that chooses a new broadcasting \timeslot immediately before  $c{^{_{\op{\timeslot}}}_{^{0}}}(x)$ must have $s_j = \bot$ in configuration $c{^{_{\op{timeslot}}}_{^{0}}}(x)$. We compute the expected value of $m_i$ in equation~$(\ref{eq:dfjpvjaurekg})$ as a function of the number of empty \timeslots, $e_i$, that $p_i$ selects from when choosing a new broadcasting \timeslot, where $e_i$ $=$ $\mid{unused\_set_i}\mid$ in configuration $c{^{_{\op{timeslot}}}_{^{0}}}(x)$.

\begin{align}
\label{eq:dfjpvjaurekg}
E\left( m_i \right) = &\\\nonumber
\sum_{t \in E_i} E\left( m_i \mid {s_i=t} \right) \Pr({s_i=t}) = &\\\nonumber
\sum_{t \in E_i}\frac{1}{ e_i } E\left( m_i \mid {s_i=t} \right) =&\\\nonumber
\sum_{t \in E_i} \frac{1}{ e_i }E \left( \sum_{{p_j \in {\cal N}_i}} 1_{\{p_j \text{ chooses \timeslot~} {t} \}}\mid {s_i=t} \right) =&\\\nonumber
\sum_{t \in E_i}\frac{1}{ e_i }\sum_{{p_j \in {\cal N}_i}} \frac{1}{\mid E_j\mid} 1_{\{ {t} \in E_j \}}1_{\{s_j=\bot\}}\nonumber
\end{align}

Our assumption that $d_i \le T-1$ implies that $ e_i >0$. Using that
$d_i=\sum_{{p_j \in {\cal N}_i}}\left(1_{\{s_j \neq \bot\}} +1_{\{s_j=\bot\}}\right)$ and, $e_i\ge T-\sum_{p_j\in {\cal N}_i} 1_{\{s_j\not = \bot\}}$,
we obtain equation~$(\ref{eq:apvkjsitekjhgls})$.

\begin{small}
\begin{align}
\label{eq:apvkjsitekjhgls}
E(m_i)\le&\sum_{{t} \in E_i} \frac{1}{T-d_i+\sum_{{p_j \in {\cal N}_i}} 1_{\{s_j=\bot\}}}\sum_{{p_j \in {\cal N}_i}} \frac{1_{\{ {t} \in E_j\}} 1_{\{s_j=\bot\}}}{\mid E_j\mid}\nonumber\\
=&\frac{1}{T-d_i+\sum_{{p_j \in {\cal N}_i}} 1_{\{s_j=\bot\}}}\sum_{{p_j \in {\cal N}_i}} \frac{1_{\{s_j=\bot\}}}{\mid E_j\mid}\underbrace{\sum_{ {t} \in E_i} 1_{\{ {t} \in E_j\}}}_{\mid E_i\bigcap E_j\mid}\nonumber \\ \le &\frac{\sum_{{p_j \in {\cal N}_i}} 1_{\{s_j=\bot\}}}{T-d_i+\sum_{{p_j \in {\cal N}_i}} 1_{\{s_j=\bot\}}}\le \frac{d_i}{T}
\end{align}
\end{small}
%
%
\end{proof}


\begin{proposition}
\label{th:ndoslo}
%

\begin{eqnarray}
\nonumber
q_i\ge\sum_{k=1}^n {\rho}_k \left( 1- \sum^{k}_{\ell=1} {\rho}_k \right)^{\frac{d_i}{T}}\\\nonumber ~~[\mathrm{clone~of~equation}~(\ref{equ:qi})]
\end{eqnarray}
\end{proposition}

\begin{proof}
Proposition~\ref{th:evidit} shows that $E({m_i}) \leq  d_i/T $.
The proposition is demonstrated by evaluating expression $(\ref{eq:cfposur})$ at $E({m_i}) = d_i/T$, see equation ~$(\ref{eq:jensen})$.
\end{proof}

Proposition~\ref{prop:2boundq} considers the concrete transmission probability $\rho_i = 1/MaxRnd$.

\noindent {\bf Proposition~\ref{prop:2boundq}}
{\em
Let ${{\rho}_i}=1/MaxRnd$. Equation~$(\ref{boundq})$ bounds from below the probability $q_i$.}

\begin{proof} In this proof, we use the letter $n$ instead of $MaxRnd$ for reason of space.
We replace ${{\rho}_i}$ with $1/n$ in equation $(\ref{equ:qi})$ to obtain equation~$(\ref{eq:fgklgfdkjl})$.
\begin{equation}
\label{eq:fgklgfdkjl}
q_i\ge \sum_{k=1}^n \frac{1}{n}\left(1-\frac{k}{n}\right)^{\frac{d_i}{T}}
\end{equation}

Equation~$(\ref{eq:psleuiruting})$ is more compact than equation~$(\ref{eq:fgklgfdkjl})$ and it is obtained by the fact that the function $(1-x)^s$ is convex.

\begin{eqnarray}
\label{eq:psleuiruting}
q_i \ge \\\nonumber \sum_{k=1}^n \frac{1}{n}\left(1-\frac{k}{n}\right)^{\frac{d_i}{T}}=\\\nonumber
\frac{1}{2n}\sum_{k=1}^n\left [\left(1-\frac{k}{n}\right)^{\frac{d_i}{T}}+\left(1-\frac{n-k+1}{n}\right)^{\frac{d_i}{T}}\right] \ge \\\nonumber
\text{ (convexity) } \frac{1}{n}\sum_{k=1}^n \left( 1-\frac{n+1}{2n}\right)^{\frac{d_i}{T}}= \\\nonumber
\left( 1-\frac{n+1}{2 n}\right)^{\frac{d_i}{T}}\nonumber
\end{eqnarray}

Another way to bound equation~$(\ref{eq:fgklgfdkjl})$ is by considering the decreasing function $y \to (1-y)^x$, as in equation~$(\ref{eq:xpaoiengnsktie})$.
\begin{align}
\label{eq:xpaoiengnsktie}
q_i&\ge\sum_{k=1}^n\frac{1}{n}\left(1-\frac{j}{n}\right)^\frac{d_i}{T}\ge\int_\frac{1}{n}^1\left(1-y\right)^\frac{d_i}{T}dy\\&=\frac{1}{\frac{d_i}{T}+1}\left(1-\frac{1}{n}\right)^{\frac{d_i}{T}+1}\nonumber
\end{align}
\end{proof}


%

\end{document}